\documentclass[article] {siamonline190516}
\usepackage[utf8]{inputenc}
\usepackage{amssymb}
\usepackage{graphics,graphicx,subfig}
\usepackage{tikz}
\usepackage{amscd}
\usepackage{bbm}
\usepackage{enumerate}
\usepackage{enumitem}

\usepackage{soul}
\usepackage{color}

\def\ee{{\mathcal E}}

\def\bbC{{\mathbb{C}}}

\def\bbG{{\mathbb{G}}}

\def\bbN{{\mathbb{N}}}

\def\bbR{{\mathbb{R}}}
\def\bbT{{\mathbb{T}}}

\def\bbZ{{\mathbb{Z}}}

\def\bf1{{\mathds{1}}}

\def\t{{\rm{T}}}
\def\SBM{{\rm{SBM}}}
\def\opr{{\rm{opr}}}

\def\bbR{{\mathbb{R}}}
\def\bbT{{\mathbb{T}}}

\usepackage[mathlines]{lineno}
\nolinenumbers

\newsiamthm{example}{Example}

 \newsiamthm{remark}{Remark}
 \newsiamthm{note}{Note}
\newtheorem{notation}{Notation}

\numberwithin{equation}{section}
\begin{document}

\title{Unified Fourier transform on graphs sampled from stochastic block models}

\author{
Mahya Ghandehari\thanks{Department of Mathematical Sciences, University of Delaware, DE, 19716, US,
\email{$\{${mahya, silo$\}$@udel.edu}}
}
\and 
Jeannette Janssen\thanks{Department of Mathematics and Statistics, Dalhousie University, NS, B3H 4R2, Canada,
\email{jeannette.janssen@dal.ca}
}
\and Silo Murphy\footnotemark[1]
}


\maketitle

\begin{abstract}
Recently, an approach to graph signal processing based on graphons was proposed.
Here we show how such a graphon-driven approach to the Fourier transform can be used on graphs sampled from a stochastic block model (SBM). In particular, we show how a Fourier basis can be easily calculated from the block sizes and the block probability matrix. 
Using perturbation theory, we derive bounds on the sensitivity of the basis with respect to variations in the block sizes.
We then consider SBMs constructed from weighted Cayley graphs. 
When block sizes are equal,  a nice Fourier basis can be derived from the representation theory of the underlying group. 
When block sizes are nearly uniform, we demonstrate that this Fourier basis closely approximates the SBM Fourier basis.
For highly non-uniform block sizes, the group-based Fourier basis is no longer applicable, though, as we show, the underlying group still provides partial information about the SBM Fourier basis.
\end{abstract}



\section{Introduction}

Graph Signal Processing (GSP) has become a popular method to process data described on irregular domains. Such data can often be interpreted as a signal that assigns numerical values to the nodes of a network or graph.
GSP has been successfully used  in applications such as network failure assessment \cite{PrakashKar2020}, detection of false data attacks in smart grids \cite{DrayerRouttenberg2020}, and analysis of brain signals \cite{HGWGBR2016,Miri2024GraphLF,MAASVB2023,miri2023brain}. For an overview of GSP and its applications, we refer to \cite{GSP-overview-2023,ortega-book,2018:Ortega:GSPOverview,SM-2013,2014:Sandryhaila:BD,SandryhailaMoura13}.

One of the most fundamental concepts in classical signal processing is the Fourier transform. This notion has been generalized to graph signals by using spectral features of the underlying graph. 
In this approach, we first assign a shift operator to the underlying graph of a signal. 
The graph Fourier transform (GFT) is then defined as the projection of  signals  onto a fixed eigenbasis of the graph shift operator. 
Using the graph Fourier transform, important concepts of signal processing have been generalized to the case of graph signals. 

This spectral approach to GSP has two major drawbacks: firstly, computing eigenbases for large matrices is computationally expensive and slow; secondly, the graph Fourier basis depends rigidly on the specific graph at hand.
However, the underlying graph of a signal may sustain minor variations due to error or changes in
the network over time.
To address these challenges, recent work has leveraged the theory of graph limits to develop an \emph{instance-independent} graph Fourier transform, which we refer to as the \emph{graphon-driven FT}.

A graphon is a symmetric measurable function on $[0,1]^2$ with values in $[0,1]$. 
Every graphon $w:[0,1]^2\to [0,1]$ gives rise to a rather general  random graph model, called the $w$-random graph, whose samples are graphs of any desired size. 
These samples share a similar large-scale structure (i.e., similar subgraph densities) governed by $w$.
The $w$-random model provides an excellent sampling scheme: almost every growing sequence generated by a $w$-random model converges to $w$ in a metric derived from the \emph{cut norm} \cite{lovasz-book}. 
This graph limit theory was initiated by Lov\'{a}sz and Szegedy in \cite{lovaszszegedy2006}.
Graphons provide a flexible framework for modeling large networks in a variety of applications (see e.g., \cite{OliRei20, DupMed22, Med14b, KVMed18,CaiHo22}).

In the graphon-driven FT, the Fourier basis is determined by the shared structure of a family of graphs, rather than by each individual graph. This shared structure is captured by the graphon that generates the family of graphs.
The graphon-driven FT offers several key benefits. 
Most importantly, instead of computing an eigenbasis for each graph sampled from a graphon, we use one Fourier basis, which is derived from the graphon.
In addition to the computational advantage, this approach allows comparison of signals across graphs with similar large-scale structure (e.g., brain networks from different individuals). 
Finally, this approach captures the approximate large-scale symmetries of graphs and translates them into the Fourier domain.
For instance, a circulant graphon has cyclic symmetries, 
yielding a graphon Fourier transform similar to the classical Fourier transform. While individual sampled graphs retain these cyclic symmetries only approximately, the graphon-driven FT based on these symmetries provides a consistent, well-structured basis applicable across all samples.
Graphon-driven FT schemes were first proposed in 2020~\cite{RuizChamonRibeiro19}, 
and the robustness of these methods was proved in \cite{Ghanehari-Janssen-Kalyaniwalla,RuizChamonRibeiro20}. 

In this paper, we develop the graphon-driven FT approach for a special class of graphon models corresponding to the \emph{stochastic block model} (SBM).
In the SBM model, vertices are divided into blocks of prescribed proportion. Edges are added independently, with probability determined by the block membership of each vertex. In a large graph sampled from an SBM, the edge density between two given blocks will be approximately equal to the link probability of those two blocks (e.g., see \cref{fig:SBM}). 
The stochastic block model is a widely used framework in network analysis with applications across several domains. SBMs were originally introduced by Holland, Laskey, and Leinhardt \cite{Holland1983} in order to provide a new framework for studying relational data in large social networks. 
Subsequent work focused on extending block-based analysis of graph relations to directed graphs \cite{NowickiSnijders2001}. 
SBMs are central to community detection and graph clustering \cite{abbe,Lee2019}, with sharp information-theoretic thresholds for exact community recovery in certain SBMs established in \cite{AbbeBandeiraHall2016}. SBMs have also been used to model inter-dependencies in infrastructure systems 
so as to understand the dynamics of failures and assist in network restoration \cite{yu_baroud_risk_and_uncertainty}. In addition, SBMs are used in information theory for graph compression, where efficient lossless algorithms with asymptotically optimal rates have been developed for many non-trivial cases \cite{Wafula_compression,graph_compression}.

\begin{figure}[!htp]
\centering
\begin{minipage}{.32\textwidth}
  \centering
\includegraphics[scale=0.15]{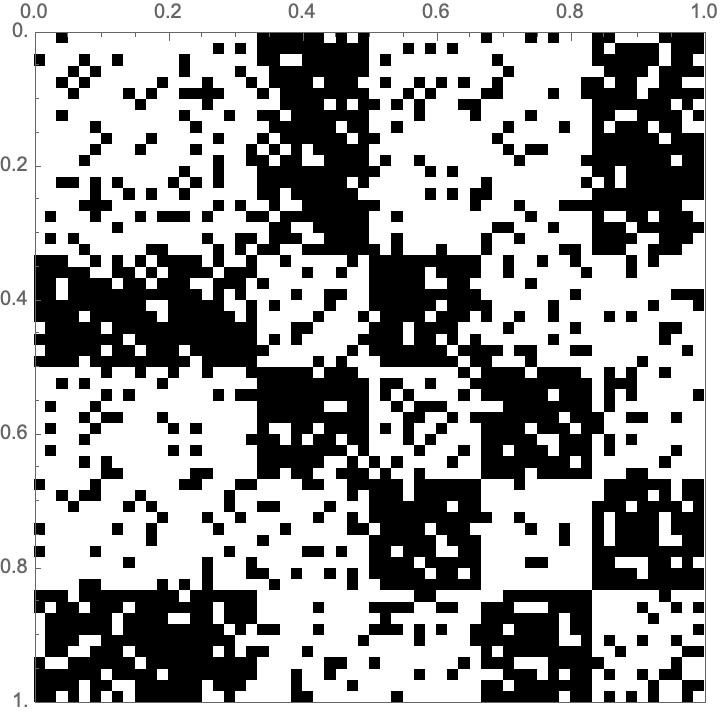}
\end{minipage}%
\begin{minipage}{.34\textwidth}
  \centering
\begin{tikzpicture}[scale=0.268]
    \draw[white] (0,0) grid (6,6);
    \draw[black] (0,0) rectangle (6,6);
    \filldraw[fill=gray!95] (0,0) rectangle (2,1);
    \filldraw[fill=gray!95] (0,3) rectangle (2,4);
    \filldraw[fill=gray!95] (2,4) rectangle (3,6);
    \filldraw[fill=gray!95] (3,3) rectangle (4,4);
    \filldraw[fill=gray!95] (0,3) rectangle (2,4);
    \filldraw[fill=gray!95] (4,0) rectangle (5,1);
    \filldraw[fill=gray!95] (2,2) rectangle (3,3);
    \filldraw[fill=gray!95] (4,2) rectangle (5,3);
    \filldraw[fill=gray!95] (5,4) rectangle (6,6);
    \filldraw[fill=gray!95] (3,1) rectangle (4,2);
    \filldraw[fill=gray!95] (5,1) rectangle (6,2);
  \end{tikzpicture}
\end{minipage}
\begin{minipage}{.32\textwidth}
  \centering
\includegraphics[scale=0.15]{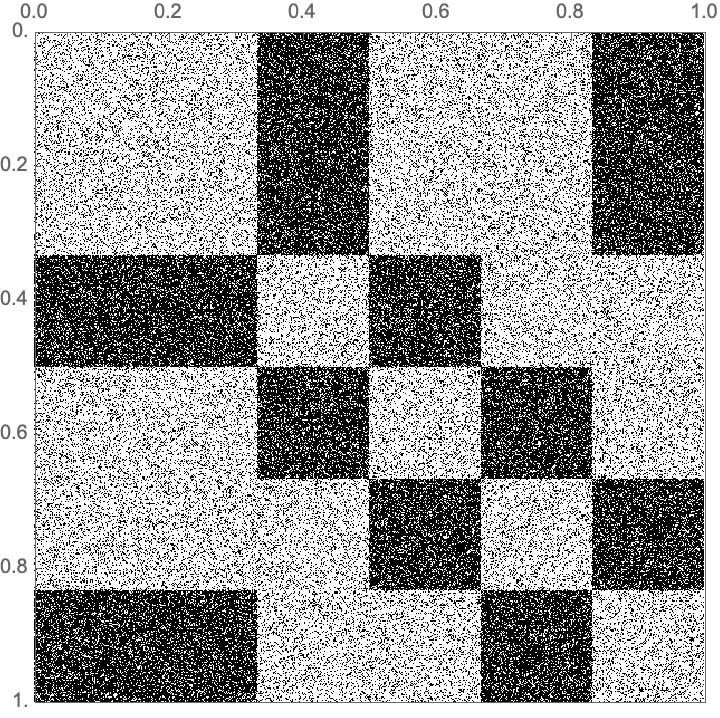}
\end{minipage}%
\caption{\scriptsize{Graphon $w$ representing SBM, shown in the middle, taking value 0.8 on gray and 0.2 on white regions. Sampled graph
 on 60 vertices (left) and one on 600 vertices (right) are shown. For graphs, the vertices have been placed on $[0,1]$, and edges are represented by black pixels.}}
\vspace{-0.27cm}
\label{fig:SBM}
\end{figure}

Interesting examples of SBMs are those with group theoretic symmetries.
Samples of such SBMs can be viewed as generalizations or perturbations of Cayley graphs. 
A Cayley graph has vertices corresponding to elements of a group and edges generated by shifts with elements of an inverse-closed subset of the group.  The underlying algebraic structure and highly symmetric nature of Cayley graphs make them a rich category of graphs for various applications. On Cayley graphs, performing GSP  compatible with the group Fourier basis has desirable properties, see \cite{BG-EUSIPCO,MDK:2019:Sampta}. 
For examples of signal processing on Cayley graphs, see \cite{permutahedron, Kotzagiannidis2018,GSP-circulant2019, Rockmore}.  
We will demonstrate that for SBM with relaxed Cayley structure, some of these favorable properties are preserved.

\subsection{Main contribution}
In this paper, we discuss how the graphon-driven Fourier transform can be used to provide suitable instance-independent Fourier transforms for samples of SBM, which we call an \emph{SBM-driven Fourier transform} for signals on such graphs. 
As a key example, we investigate the case of SBM with group theoretic symmetries. 
Following the approach in \cite{RuizChamonRibeiro20}, we use the adjacency matrix (instead of the Laplacian matrix) as our shift operator for developing the desired GFT scheme.

Our contribution is three-fold. Firstly, we show how the problem of computing the eigenspace decomposition of an SBM can be turned into a low-dimensional problem, and how the group Fourier basis can be used to find a Fourier basis for SBM with Cayley structure and uniform block sizes. Secondly, we show that relatively small changes in the block sizes
do not affect the obtained eigen-decomposition dramatically. 
Thirdly, we find a basis for stochastic block models with an underlying Cayley structure, even if the symmetry is broken due to unequal block sizes.
In each scenario, the Fourier basis associated with the SBM defines a SBM-driven Fourier transform for graphs generated from the SBM.

\subsection{Organization of the paper}
In \cref{sec:notation}, we collect notations and background regarding the stochastic block model, graphons, and graph/graphon Fourier transforms. We end this section by quoting \cref{thm:convergence}, which lays the groundwork for developing graphon-driven GSP.
In \cref{sec:main}, we discuss the special case of the graphon-driven Fourier transform for samples of an SBM.
Namely, we present a method to construct Fourier basis for an SBM using the model matrix and probability matrix associated with the SBM (\cref{subsec:computing-basis}), and show how it yields an SBM-driven Fourier transform for sampled graphs (\cref{Appendix-sec:implement}).
We use this method in \cref{sec:cayley-general} to construct Fourier basis for SBM with a Cayley structure.
In particular, we construct the SBM Fourier basis in the two cases of equal block sizes (\cref{prop:GSP-Cayley-uniform}) and arbitrary block sizes (\cref{thm: general-decomp}).
To establish the robustness of the Fourier transform from \cref{prop:GSP-Cayley-uniform}, we dedicate \cref{subsec:perturbations} to investigate the effect of the block size perturbations on the SBM Fourier transform.
In \cref{thm: general-decomp}, we apply these results to the case of a Cayley SBM with Abelian underlying group. It turns out that non-uniformity of block sizes can be interpreted as weighted Fourier transform (\cref{remark-waighted}).
We conclude the paper with an example on SBMs with $\bbZ_5$ symmetries in \cref{sec:experiment}.

\section{Notations and background}
\label{sec:notation}
In this section, we collect the relevant information regarding the stochastic block model, graphons, and graph/graphon Fourier transforms. 

For any positive integer $n$, let $[n]$ denote the set $\{1,\ldots,n\}$. For integers $n,m$, we use the notation $[n,m]$ to denote the set $\{n,\ldots,m\}$.
We think of $\bbC^n$ as the vector space, equipped with the inner product $\langle x,y\rangle=\sum_{i=1}^n x_i\overline{y_i}$. Similarly, $L^2[0,1]$ denotes the vector space of square-integrable, measurable functions equipped with inner product $\langle f,g\rangle_{L^2[0,1]}=\int f(x) \overline{g(x)}dx$. We use $\|\cdot\|_\opr$ to denote the operator norm of a matrix acting between the appropriate $L^2$ spaces.
Throughout the paper, uppercase letters denote matrices, while lowercase letters (e.g., $x,y,$ etc.) denote vectors.

\subsection{Graphons and the stochastic block model}
\label{sec:SBM}

The \emph{stochastic block model (SBM)} is a random graph model defined by a sequence of \emph{block sizes} $k_1, \ldots, k_n\in \mathbb{N}$ and  a symmetric $n\times n$ \emph{probability matrix} $A$. In this model, graphs are generated as follows:  
Initially, $m=\sum_{i=1}^n k_i$ vertices are created and partitioned into sets $B_1,\dots,B_n$, referred to as \emph{blocks}, according to the prescribed block sizes; 
then, for each $i,j\in [n]$, and for all $u\in B_i$, $v\in B_j$, edges $uv$ are independently added with probability $a_{i,j}$.
In order to create converging sequences of block model graphs, we will 
fix a probability measure $\mu $ on $[n]$, representing the fraction of vertices contained in each block. 
For a given $n\times n$ matrix $A$ and measure $\mu$ on $[n]$, we can then consider a sequence $\{G_N\}_N$ of graphs of vertex size $N$  formed according to the stochastic block model defined by probability matrix $A$ and block sizes $k_i=\mu_i N$ ($1\leq i\leq n$). To simplify the exposition, we assume that $\mu_iN$ is an integer for every $i$.  We will say that {$G_N\sim \SBM(A,\mu,N)$}. 

\begin{remark}
Our exposition can easily be extended to remove the restriction that $\mu_iN$ is an integer.  Namely, we can take $\{k_i\}_{i=1}^n$ so that $|k_i-\mu_iN|<1$ for all $i\in [n]$ and $\sum_{i=1}^n k_i=N$. The asymptotic results in this paper will remain true.
\end{remark}

Recall that a graphon is a measurable symmetric function $w:[0,1]^2\to [0,1]$. 
The appropriate norm on the space of graphons is the \emph{cut norm},  introduced in \cite{cut-norm} and denoted by $\|\cdot\|_\Box$. For an integrable function $f:[0,1]^2\to \bbR$, the cut norm is defined as follows:
$$\|f\|_\Box:=\sup_{S,T\subseteq [0,1]}\left|\iint_{S\times T}f(x,y)\, dx\, dy\right|,$$
where $S$ and $T$ are taken over all measurable subsets of $[0,1]$.

We can represent matrices and graphs by graphons as described below.
\begin{definition}[Graphs and matrices as graphons]
\label{exp-graphons1}
         An $N\times N$ matrix $A=[a_{i,j}]$ can be represented as a graphon $w_A$ by setting $w_A$ equal to $a_{i,j}$ on $[\frac{i-1}{N}, \frac{i}{N})\times [\frac{j-1}{N}, \frac{j}{N})$ for every $i,j\in [N]$.
         Any labeled graph $G$ can be represented by the ($\{0,1\}$-valued) graphon, denoted by $w_G$, corresponding to its adjacency matrix.
\end{definition}
\begin{definition}[SBMs as graphons]
\label{exp-graphons2}
        An SBM given by $n\times n$ probability matrix $A$ and measure $\mu$ on $[n]$ can be presented as a graphon $w_{A,\mu}$ defined as follows. Let $\{I_j\}_{j=1}^n$ be a partition of $[0,1]$ into consecutive intervals so that, for all $j$, $|I_j|=\mu_j$. Then for all $1\leq i\leq j\leq n$ and for all $s\in I_i$, $t\in I_j$, $$w_{A,\mu}(s,t)=w_{A,\mu}(t,s)=a_{i,j}.$$ 
\end{definition}

Let $A$ and $\mu$ be as in \cref{exp-graphons2}.
A sequence $\{ G_N\}_N$, when $G_N\sim {\SBM}(A,\mu, N)$ and $N\to \infty$,  almost surely forms a convergent graph sequence in the sense of graph limit theory (see \cite{lovaszszegedy2006}). Precisely, with suitable labelings of the graphs $\{G_N\}_N$, the sequence  $\{ w_{G_N}\}_N $ of associated graphons converges in cut norm to the graphon $w_{A,\mu}$.

Sometimes it will be useful to represent $w_{A,\mu}$ in the form of a matrix of a certain size. This motivates our introduction of the model matrix $W$.
\begin{definition}[Model matrix]\label{def:model-matrix}
 For $i\in [n]$, let $k_i=\mu_iN$. For any two positive integers $s,t$, let $J_{s,t}$ denote the $s \times t$ matrix consisting entirely of ones. 
Let
$W=[a_{i,j}J_{k_i,k_j}]$ be the matrix obtained by replacing every entry $a_{i,j}$ in $A$ by a block $a_{i,j}J_{k_i,k_j}$. 
We call  $W$ the \emph{model matrix} of $\SBM(A,\mu,N) $. For any two vertices $u$ and $v$ in  $G_N\sim {\SBM}(A,\mu, N)$, the probability that $u$ and $v$ are adjacent is given by $W_{u,v}$. 
\end{definition}

\subsection{The graph Fourier transform}
\label{sec:graphFT}
For a fixed graph $G$ with vertex set $V(G)$, a \emph{graph signal} on $G$ is a function $f:V(G)\to {\mathbb C}$. If the vertex set $V(G)$ is labeled, say $\{v_i\}_{i=1}^N$, then the graph signal can be represented as a column vector $\left[f(v_1), f(v_2),\ldots, f(v_N)\right]^{\rm T}$ in ${\mathbb C}^N$, where ${\rm T}$ denotes the matrix transpose.  

To define the \emph{graph Fourier transform}, we first assign a shift operator to the underlying graph.
The transform is then given by the expansion of signals onto a fixed eigenbasis of this graph shift operator.
Common choices for the shift operator include the graph adjacency matrix and the graph Laplacian. 
The adjacency matrix of a graph $G$ with $N$ nodes is a $0/1$-valued matrix $A_G$ of size $N$, whose $(i,j)$-th entry is 1 precisely when the vertices $v_i$ and $v_j$ are adjacent. The Laplacian of $G$, denoted by $L_G$, is the $N\times N$ matrix
given by $L_G= D_G-A_G$, where $D_G$ is the diagonal matrix with entries $d_{ii}$ equal to the degree of vertex $v_i$.
The selection of the graph shift operator significantly influences the properties of the resulting graph Fourier transform (see \cite[Chapter 3]{ortega-book}). For an overview of various shift operators used for developing graph Fourier transform, see \cite{SandryhailaMoura13} and references therein.

In this paper, we take the adjacency matrix as our graph shift operator. Our choice is due to the fact that the spectral features of the adjacency matrices associated with a converging sequence of graphs converge in an appropriate sense (see \cite{szegedy-spectra}). This phenomenon allows us to leverage the graph limit theory to produce a consistent graphon-driven Fourier transform. 
Our choice of shift operator matches the choice in the definition of graphon Fourier transform proposed in \cite{RuizChamonRibeiro20} or the graphon neural networks in \cite{neural1}.

Being real symmetric matrices, adjacency matrices are unitarily diagonalizable. Let $A_G=U^*\Lambda U$, where $U$ is a unitary matrix and $\Lambda$ is the diagonal matrix whose diagonal entries are the eigenvalues of $A_G$. Given a graph signal $x\in \bbC^N$ on $G$, the graph Fourier transform of $x$ is defined as
    \begin{equation}
        \label{eq:GFT}
        \widehat{x}= Ux.
    \end{equation}
    Given $\widehat{x}$, we can retrieve the original signal via the inverse Fourier transform defined as
    \begin{equation}
        x=U^* \widehat{x}.
    \end{equation}
 For a thorough discussion on graph Fourier transform and its applications, see for example~\cite{2018:Ortega:GSPOverview,ortega-book}.
\subsection{The graphon Fourier transform}
\label{sec:graphonFT}
In this section, we introduce the graphon Fourier transform, a transform on signals in $L^2[0,1]$ that serves as the foundation for the graphon-driven Fourier transform on graph signals.
The viability of the graphon-driven approach is based on \Cref{thm:convergence} which establishes the convergence of the graph FT to the graphon FT. This result is central to our approach, so we give a precise statement of the theorem. In order to do so, we need to introduce the notations necessary for graphon FT.

Similar to the graph Fourier transform, the \emph{graphon Fourier transform} is defined through the spectral decomposition of the related graphon, or to be precise, the spectral decomposition of the associated operator.
Every graphon $w:[0,1]^2\to [0,1]$ acts as the kernel of an integral operator on the Hilbert space $L^2[0,1]$ as follows:
$$T_w: L^2[0,1]\to L^2[0,1], \quad T_w(\xi)(x)=\int_{[0,1]} w(x,y)\xi(y)\, dy, \mbox{ for } \xi\in L^2[0,1], x\in [0,1].$$ 
An $L^2$-function is a $\lambda$-eigenfunction of $T_w$ (or $\lambda$-eigenfunction of $w$ for short) if
$T_wf=\lambda f$ almost everywhere. 
For any graphon $w$, $T_w$ is a self-adjoint compact operator, and its operator norm is bounded by 1. Thus, $T_w$ has a countable spectrum lying in the interval $[-1,1]$ for which 0 is the only possible accumulation point. 
Let $n^+$ and $n^-$ in ${\mathbb N}_*:={\mathbb N}\cup \{\infty\}$ denote the number of positive and negative eigenvalues of $T_w$ respectively. 
We label the nonzero eigenvalues of $T_w$ as follows:
\begin{eqnarray}
\label{eq:ordering}
1\geq \lambda_1(w)\geq \lambda_2(w)\geq \ldots  >0 \ \mbox{ and } \ 0> \ldots\geq \lambda_{-2}(w)\geq \lambda_{-1}(w)\geq -1
\end{eqnarray}
Note that if $n^+$ (resp.~$n^-$) is finite, the corresponding chain of inequalities is a finite chain. 

From the spectral theory for compact operators, we have that $L^2[0,1]$ admits an orthonormal basis containing eigenfunctions of $T_w$. Let $I_w\subseteq {\mathbb Z}\setminus\{0\}$ be the indices in \eqref{eq:ordering} enumerating the nonzero (repeated) eigenvalues of $T_w$. Let $\{\phi_i\}_{i\in I_w}$ be an orthonormal collection of associated eigenfunctions. Then the spectral decomposition of $T_w$ is given as follows.
\begin{equation}\label{eq:spectral-decom}
T_w=\sum_{i\in I_w}\lambda_i(w) \, \phi_i\otimes \phi_i,
\end{equation}
where $\phi_i\otimes \phi_i$ denotes the rank-one projection  defined as $(\phi_i\otimes\phi_i)(\xi)=\langle\xi,\phi_i\rangle\phi_i$ for every $\xi\in L^2[0,1]$.
Clearly, $\{\phi_i\}_{i\in I_w}$ forms an orthonormal basis for the image of $T_w$.

We can now define graphon signals and the graphon Fourier transform.
The concept of graphon Fourier transform was first introduced in \cite{RuizChamonRibeiro20} for graphons without any repeated eigenvalues. We quote the definition from a slightly generalized version proposed in \cite{Ghanehari-Janssen-Kalyaniwalla}, where the condition on eigenvalues was dropped.
\begin{definition}[Graphon Fourier transform]\cite[Definition 3.4]{Ghanehari-Janssen-Kalyaniwalla}
\label{def:graphonFT}
Let $w:[0,1]^2\to [0,1]$ be a graphon, and consider the spectral decomposition of the associated integral operator as in \eqref{eq:spectral-decom}.
\begin{itemize}
\item[(i)] Any square-integrable measurable function $f\in L^2[0,1]$ is called a \emph{graphon signal} on $w$.
\item[(ii)] For any distinct eigenvalue $\lambda$ of $T_w$, let $I_\lambda(w)$ denote the set of indices $i$ so that $\lambda_i(w)=\lambda$.
\item[(iii)] The \emph{graphon Fourier transform} of a signal $f$ is the projection of $f$ onto eigen-spaces of $T_w$. More precisely, the \emph{Fourier projection} of $f$ corresponding to $\lambda$, denoted by $\widehat{f}(\lambda)$, is defined as
\begin{equation}\label{eqn:FT_Proj}
\widehat{f}(\lambda) = \sum_{i\in I_\lambda(w)} (\phi_i\otimes\phi_i)(f)=\sum_{i\in I_\lambda(w)} \langle f,\phi_i\rangle \,\phi_i.
\end{equation}
In addition, $\widehat{f}(0)$ is the projection of the signal onto the null space of $T_w$.
\end{itemize}
\end{definition}
In order to state the convergence result that motivates this work,
we need to view graph signals as elements of the space $L^2[0,1]$ of graphon signals.
\begin{definition}\label{def:vector-to-function}
Let $x$ be  a signal on a graph $G$ with labeled vertex set $[N]$. We associate the step-function $f_{x}\in L^2[0,1]$ to the graph signal $x$ defined as:
\begin{equation}\label{eq:vectorstofunctions}
f_{x}(t)={\sqrt{N}}x_j\text{ for all }\ t\in \left[\frac{j-1}{N},\frac{j}{N}\right],\, j\in [N].
\end{equation}

\end{definition}
Here $\sqrt{N}$ is the scaling factor which ensures that $\|f_{x}\|_{L^2[0,1]}=\|x\|_{\bbC^N}$.

\begin{remark}\label{rem:eigen-matrix-graphon}
The spectral features of an $N\times N$ matrix $A=[a_{i,j}]$ and its associated graphon $w_A$ (\cref{exp-graphons1}) are closely related. 
Let $x\in \bbC^N$ be a vector, and let $f_x$ denote the representation of $x$ as a function in $L^2[0,1]$ as given in \eqref{eq:vectorstofunctions}. It is easy to verify that 
$(Ax)_i=\sqrt{N}(T_{w_A}f_x)(t)\text{ for all }t\in I_i.$
Consequently, for $\lambda\neq 0$, the map $x\mapsto f_x$ forms a one-to-one correspondence between $\lambda N$-eigenvectors of $A$ and $\lambda$-eigenfunctions of $T_{w_A}$.
\end{remark}

Let $\{G_n\}_n$ be a sequence of labeled graphs of increasing size converging to the graphon $w$ in cut norm.
For each $n$, let $\{\lambda_i{(G_n)}\}_{i\in {\mathbb Z}}$ be the sequence of eigenvalues of the adjacency matrix of $G_n$ ordered as in \eqref{eq:ordering}, and padded by 0 to create an infinite sequence.  
By \cite[Theorem 11.54]{lovasz-book},  we have 
$$\lim_{n\to \infty}\frac{\lambda_i(G_n)}{|V(G_n)|}=\left\{\begin{array}{cc}
   \lambda_i(w)  & \mbox{ if } i\in I_w \\
  0   & \mbox{otherwise}
\end{array}\right.$$

We can now formulate the convergence of the GFT to the graphon Fourier transform.
The following theorem was first proved in \cite{RuizChamonRibeiro20} for graphons that do not admit any repeated eigenvalues.

\begin{theorem}(\cite[Theorem 3.7]{Ghanehari-Janssen-Kalyaniwalla})\label{thm:convergence}
Let $\{G_n\}_n$ be a graph sequence converging to a graphon $w$ in cut norm, and $\{x_n\}_n$ be a sequence of graph signals on each of the $G_n$, such that the corresponding sequence of step-functions $\{f_{{x}_n}\}_n$ converges to a graphon signal $f$ in $L^2[0,1]$. With the notation introduced in this section, we  have that for each distinct nonzero eigenvalue $\lambda$ of $T_w$,
\begin{equation}\label{eq:converg}
 \left(\sum_{i\in I_\lambda(w)}
 \phi^n_i\otimes \phi^n_i\right) ({f_{x_n}})\rightarrow  
 \widehat{f}(\lambda)\mbox{ in } L^2[0,1] \mbox{ as }  n\rightarrow \infty,
\end{equation}
where each $\phi^n_i$ is an eigenvector of $G_n$ corresponding to $\lambda_i(G_n)$, represented as a step-function.
Recall that $I_\lambda(w)$ is the set of indices $j$ such that $\lambda_j(w)=\lambda$.
\end{theorem}

The graph Fourier transform projects a graph signal onto an eigenbasis of the graph adjacency matrix. \Cref{thm:convergence} states that the GFT on a graph sampled from a graphon $w$ can be approximated by projection onto the appropriate eigenspaces of $w$. 
The graphon Fourier transform can inform a graphon-driven FT applicable to graph signals. 
We give a rather general example below. 
\begin{example}[A proposed graphon-driven FT]
Let $w$ be a graphon, and suppose that every eigenvalue of $w$ is simple. 
Let  $\{\phi_i\}_{i}$ be an orthonormal collection of eigenfunctions associated with nonzero eigenvalues of $T_w$.
Fix $\delta>0$, and define $J_\delta:=\{i\in {\mathbb Z}: |\lambda_i(w)|>\delta\}$ 
and $N:=|J_\delta|$. 

Let $G$ be a graph on $N$ vertices sampled from $w$.
For a signal $x\in\bbC^N$ on $G$, the graphon-driven FT is defined as follows:
$$\widehat{x}(\lambda_i)=\langle f_x,\phi_i\rangle, \ \mbox{ for } i\in J_\delta,$$
where $f_x\in L_2[0,1]$ is the associated step function to $x$, as in Definition \ref{def:vector-to-function}.

Since $G$ is sampled from $w$, we know that $\|w_G-w\|_\Box$ is small. 
\Cref{thm:convergence} implies that the $i$-th Fourier coefficient of $x$ (when the precise graph Fourier transform obtained from the eigen-decomposition of $A_G$ is applied) can be approximated with the graphon Fourier coefficient $\widehat{f_x}(\lambda)$. 
Note that the graphon-driven Fourier transform only depends on the graphon, and not on the particular sample $G$ of $w$. 
That is, for different samples of $w$ we do not need to calculate the precise eigen-decomposition for each adjacency matrix.
\end{example}

\begin{remark}
The situation is more subtle when eigenvalues of $w$ have higher multiplicity; see, for example,~\cite[Example 3.11]{Ghanehari-Janssen-Kalyaniwalla}. 
In this case, \cref{thm:convergence} suggests that the graphon-driven FT should be formulated in terms of  projections onto the space spanned by eigenvectors associated with eigenvalues that converge to the same limit.
For instance, suppose that $\lambda$ is an eigenvalue of $w$ with multiplicity $2$, so that $\lambda=\lambda_1(w)=\lambda_2(w)$. Then the sequences
\[
\left\{ \frac{\lambda_1(G_n)}{|V(G_n)|}\right\}_n
\quad \text{and} \quad
\left\{\frac{\lambda_2(G_n)}{|V(G_n)|}\right\}_n
\]
both converge to $\lambda$. In general, however, the individual GFT coefficients associated with $\lambda_1(G_n)$ and $\lambda_2(G_n)$ need not converge. Rather, \cref{thm:convergence} states that the projection of a signal onto the space spanned by the eigenvectors of $G_n$ corresponding to $\lambda_1(G_n)$ and $\lambda_2(G_n)$ converges to the projection of that signal onto the $\lambda$-eigenspace of $T_w$. 
\end{remark}

A detailed discussion on graphon-driven FT is beyond the scope of this paper. 
Instead, we focus on the special case of stochastic block models in this paper.

\section{The SBM-driven Fourier transform}\label{sec:main}
In this section, we particularize the concept of the graphon-driven Fourier transform to the case of stochastic block models taking the associated model matrix into account. The graphon-driven Fourier transform specialized to the setting of SBMs is referred to as the \emph{SBM-driven Fourier transform}.

Consider a sequence $\{ G_N\}_N$ of graphs of increasing size, where each $G_N$ is sampled from a stochastic block model with probability matrix $A $ and measure $\mu$. As noted earlier, the sequence $\{ G_N\}_N$ converges almost surely to the graphon $w_{A,\mu}$ associated with the SBM.
Motivated by the graphon-driven FT, we define the SBM-driven Fourier transform through the eigenbasis of $w_{A,\mu}$.
In this section, we discuss how to obtain this basis and analyze its sensitivity to changes in the block sizes. 

\subsection{SBM and model matrix}
\label{subsec:GraphonModelMatrix}
Consider a stochastic block model $\SBM (A,\mu,N)$, and assume that $\mu_iN\in \bbN$ for all $i\in [n]$.
Under this assumption, we have that $w_{A,\mu}$ is the graphon representation of the model matrix $W$ 
(see \cref{def:model-matrix} for the definition of the model matrix and \cref{exp-graphons1} for the definition of the graphon representation of a matrix).
Given the correspondence between eigenspaces of $W$ and $T_{w_{A,\mu}}$ (\cref{rem:eigen-matrix-graphon}), the SBM-driven Fourier transform will be defined via  an eigenbasis of the model matrix $W$ as follows.

\begin{definition}[SBM-driven Fourier transform]
    Suppose a graph $G$ is sampled from the stochastic block model $\SBM (A,\mu,N)$, and let $x\in \bbC^N$ be a signal on $G$. For an eigenvalue  $\lambda$ of $W$, define the SBM-driven Fourier transform as follows:
    $$\widehat{x}(\lambda) =\left\{\begin{array}{cc}
      P_\lambda (x)   &  \lambda\neq 0\\
       x-\sum_{\lambda\in \sigma(W), \lambda\not=0}P_\lambda(x),  & \lambda=0
    \end{array} \right.,$$
    where $P_{\lambda}$ is the projection, on $\bbC^N$, onto the $\lambda$-eigenspace of $W$.
    Here $\sigma(W)$ denotes the set of eigenvalues of $W$.
\end{definition}

In practice, and in accordance with experimental observations, eigenvectors corresponding to large eigenvalues 
are often regarded as smooth (when the graph Fourier transform is defined via the adjacency matrix), while eigenvectors corresponding to large negative eigenvalues reveal contrast.  Signal compression is achieved by retaining only the Fourier coefficients associated with eigenvalues of large magnitude \cite[Section VI]{SM-2013}. Intuitively, this can be understood from the fact that the shift operator is typically interpreted as describing the evolution of a signal over time. Accordingly, eigenvectors associated with eigenvalues of small magnitude correspond to components that tend to dissipate.

In the graphon-driven Fourier approach, the Fourier projection of $x$ corresponding to eigenvalue 0 is the projection of $f_x$ onto the kernel of $T_{w_{A,\mu}}$. 
However, contrary to other eigenspaces, the kernel of $T_{w_{A,\mu}}$ and of $W$ do not correspond exactly.
It follows easily from the definition of $w_{A,\mu}$ (and also from \cref{prop:lin-alg-compression} below) that $T_{w_{A,\mu}}$ has rank at most $n$.
Therefore, $T_{w_{A,\mu}}$ has eigenvalue zero with infinite multiplicity. The rank of $W$ equals that of $T_{w_{A,\mu}}$, so the dimension of the kernel of $W$ is finite but grows with $N$, more precisely, we have $N-n\leq \dim(\ker{W})\leq N$.

\subsection{Computing the SBM Fourier basis}\label{subsec:computing-basis}
In this section we will see  that the eigen-decomposition of the model matrix $W$  is closely related to that of a much smaller matrix $A_\mu$, which we now define.
\begin{definition}\label{def:M-A-mu-D}
Let $\SBM(A,\mu,N)$ be an SBM with an $n\times n$ probability matrix $A$. We introduce the following matrices.
\begin{enumerate}[label=(\alph*)]
\item\label{eq:M-weight} The diagonal matrix $M:={\rm diag}(\mu_1,\ldots,\mu_n)$ is called the \emph{weight matrix}.
\item\label{A-mu} The $n\times n$ matrix $A_\mu := \sqrt{M}A\sqrt{M}$ is called the \emph{weighted probability matrix}.
\item\label{def:D} The $N \times n$ matrix $D$ is defined in block form as
    $D={\rm diag}(J_{k_1,1}, J_{k_2,1},\ldots, J_{k_n,1})$,
    where $J_{k_i,1}$ is the column vector in ${\mathbb C}^{k_i}$ whose entries are all ones.
\item\label{def:V-isom} The $N\times n$ matrix $V$ is defined as 
$V:=\frac{1}{\sqrt{N}}DM^{-\frac{1}{2}}.$
\end{enumerate}
\end{definition}
Note that for any vector $x\in \bbC^n$, $Dx$ is the \lq\lq blow-up\rq\rq  vector in {the high-dimensional space} $\bbC^N$ obtained from $x$ by repeating each entry $x_i$ exactly $k_i$ times.

\begin{lemma}\label{lem:relations-between-matrices}
With notations as given above in 
\cref{def:M-A-mu-D}, and model matrix $W$ as in \cref{def:model-matrix}, we have the following.
\begin{enumerate}[label=(\alph*)]
\item\label{part-lem:W=DADt-and DtD=NM} $W = DAD^\t$ and $\frac{1}{N} D^\t D = M$.
\item\label{lem-part:DM^{-1/2} is an isometry}  The operator $V$ is an isometry, i.e., $V^\t V=I$.
\item\label{lem-part:VA-muV-W} $VA_\mu V^\t=\frac{1}{N}W$
and $A_\mu =\frac{1}{N}(V^\t WV)$.
\item\label{lem-part:D-norm} $\|D\|_\opr=\sqrt{N\max\{\mu_1,\ldots,\mu_n\}}$.
\end{enumerate}
\end{lemma} 
\begin{proof}
\cref{part-lem:W=DADt-and DtD=NM} and the first equation in \cref{lem-part:VA-muV-W} follow directly from the definitions. 
To prove \cref{lem-part:DM^{-1/2} is an isometry}, 
   let $x,y\in \bbC^n$ be given. Since $D^\t=D^*$ and $M^{-\frac{1}{2}}=(M^{-\frac{1}{2}})^*$, we get
   \begin{eqnarray*}
        \langle Vx, Vy\rangle_{\bbC^N} &=& \langle \frac{1}{\sqrt{N}}DM^{-1/2}x, \frac{1}{\sqrt{N}}DM^{-1/2}y\rangle_{\bbC^N}
       = \langle \frac{1}{N}D^\t DM^{-1/2}x, M^{-1/2}y\rangle_{\bbC^n}\\
       &=& \langle M^{-1/2}MM^{-1/2}x, y \rangle_{\bbC^n}
       =\langle x,y\rangle_{\bbC^n}.
   \end{eqnarray*}
   Thus, $V$ is an isometry, i.e., $V^*V=I$, which implies that $V^\t V=I$ since $V$ is real-valued. 
   The first equation of \cref{lem-part:VA-muV-W}, together with \cref{lem-part:DM^{-1/2} is an isometry}, immediately proves the second equation of \cref{lem-part:VA-muV-W}.
   Finally, \cref{lem-part:D-norm} can be obtained from \cref{part-lem:W=DADt-and DtD=NM} as follows:
   $$\|D\|^2_\opr=\|D^\t D\|_\opr=N\|M\|_\opr=N\max_{1\leq i\leq n}\mu_i.$$
\end{proof}

It follows from \cref{lem:relations-between-matrices} that $W$ is a multiple of the conjugation of $A_\mu$ by an isometry. 
We will now see that the spectral behavior of $A_\mu$ and $W$ are closely related. 
\begin{proposition}\label{prop:lin-alg-compression}
Let $W$,  $M$ and $A_\mu$ be the model matrix, the weight matrix, and the weighted probability matrix of an SBM defined by $(A,\mu,N)$, and let $D$ and $V$ be as given in \cref{def:M-A-mu-D}.
Let $\lambda\in {\mathbb R}$,  and $x\in {\mathbb C}^n$ and $y\in \bbC^N$ be unit vectors.
Then, 
\begin{enumerate}[label=(\alph*)]
\item\label{part-prop:lin-alg:Y-eigen->X-eigen} Suppose $\lambda\neq 0$.
If $y$ is a unit $\lambda$-eigenvector of $W$ then $y\in {\rm range}(V)$, and  $x=V^\t y$ is a unit
$\frac{\lambda}{N}$-eigenvector of $A_\mu$.
\item\label{part-prop:lin-alg:X-eigen->Y-eigen} If $x$ is a unit
$\frac{\lambda}{N}$-eigenvector of $A_\mu$ then $y=V x$
is a unit $\lambda$-eigenvector of $W$.
\item\label{part-prop:lin-alg:multiplicity}
$\lambda$ is a nonzero eigenvalue of $W$ of multiplicity $t$ if and only if $\frac{\lambda}{N}$ is a nonzero eigenvalue of $A_\mu$ of multiplicity $t$.
\end{enumerate}
\end{proposition}
\begin{proof}
       To prove \cref{part-prop:lin-alg:Y-eigen->X-eigen}, suppose that $y$ is a unit $\lambda$-eigenvector of $W$, i.e, $Wy = \lambda y$. 
    Since $W = N(VA_\mu V^\t)$ (by \cref{lem:relations-between-matrices} \cref{lem-part:VA-muV-W}) and $y=\frac{1}{\lambda} Wy$, we get $y\in {\rm range}(V)$. Consequently, since $V^\t$ restricted to ${\rm range}(V)$ is an isometry, the vector $x=V^\t y$ is a unit vector as well. More precisely, since  $y=V\tilde{y}$ for some $\tilde{y}\in\bbC^n$ {and $V$ is an isometry}, we have
    $$\|x\|=\|V^\t y\|=\|V^\t V\tilde{y}\|=\|\tilde{y}\|=\|V\tilde{y}\|=\|y\|=1.$$
    Using \cref{lem:relations-between-matrices} \cref{lem-part:VA-muV-W}, we have
        $$\frac{1}{N}V^\t Wy = \frac{1}{N}{V^\t}(N VA_\mu V^\t)y=A_\mu V^\t y=A_\mu x.$$
    Since $y$ is a $\lambda$-eigenvector of $W$, we have $\frac{1}{N}V^\t Wy = \frac{\lambda}{N}V^\t y=\frac{\lambda}{N}x$. 
    Putting these two identities together, we get that $x$ is a unit $\frac{\lambda}{N}$-eigenvector of $A_\mu$. 

    To prove \cref{part-prop:lin-alg:X-eigen->Y-eigen}, assume that $x$ is a unit vector satisfying 
    $A_\mu x=\frac{\lambda}{N}x$. Since $V$ is an isometry, $y=Vx$ is a unit vector as well. Moreover, using \cref{lem:relations-between-matrices} \cref{lem-part:VA-muV-W} again, we have
    $$Wy=(N V A_\mu V^\t)(Vx)=N V A_\mu x=NV(\frac{\lambda}{N}x)=\lambda y.$$

To prove \cref{part-prop:lin-alg:multiplicity}, suppose $\lambda\neq 0$ is an eigenvalue of $W$, and $\{y_1,\ldots,y_t\}$ is an orthonormal basis for the associated $\lambda$-eigenspace. 
By \cref{part-prop:lin-alg:Y-eigen->X-eigen}, the set $\{V^\t y_1,\ldots,V^\t y_t\}$ is a subset of the $\frac{\lambda}{N}$-eigenspace of $A_\mu$. Moreover, since each $y_i$ belongs to the range of $V$ (by \cref{part-prop:lin-alg:Y-eigen->X-eigen}), and using the fact that the operator $V^\t$ when restricted to the subspace ${\rm range}(V)$ is an isometry, we observe that the set $\{V^\t y_1,\ldots,V^\t y_t\}$ is orthonormal as well. 
Thus the multiplicity of $\frac{\lambda}{N}$ as an eigenvalue of $A_\mu$ is at least $t$.
Conversely, let $\frac{\lambda}{N}$ be a nonzero eigenvalue of $A_\mu$ with  orthonormal eigenbasis $\{x_1,\ldots,x_s\}$.
By \cref{part-prop:lin-alg:X-eigen->Y-eigen}, the set $\{V x_1,\ldots,Vx_s\}$ is an orthogonal family of $\lambda$-eigenvectors of $W$. Therefore the multiplicity of $\frac{\lambda}{N}$ as an eigenvalue of $A_\mu$, is at most $t$.
\end{proof}
With this proposition in hand, we can reduce computing the eigen-decomposition of the range of the matrix $W$ to the eigen-decomposition of the smaller matrix $A_\mu$.
Note that $W$ is  $N\times N$ (and $N$ usually tends to infinity), whereas the size of $A_\mu$ is equal to the number of the blocks in SBM.
This correspondence extends to the nonzero eigenvalues/eigenvectors of the graphon operator $T_{w_{A,\mu}}$ as well, since
 $\lambda\neq 0$ is an eigenvalue of $T_{w_{A,\mu}}$ if and only if $\lambda N$ is an eigenvalue of $W$ (\cref{rem:eigen-matrix-graphon}). The following corollary then follows directly from \cref{prop:lin-alg-compression}.

\begin{corollary}
\label{cor:graphon_eigenvalues_SBM}
The nonzero eigenvalues of the graphon $w_{A,\mu}$  are exactly the nonzero eigenvalues of $A_\mu$.
\end{corollary}
%

\subsection{Implementing the SBM-driven Fourier transform}
\label{Appendix-sec:implement}
Let $G$ be a graph sampled from $\SBM(A,\mu,N)$ and let $x$ be a signal on $G$.
In this section, we show exactly how to compute the SBM-driven Fourier transform of $x$; a summary of the steps is given in \cref{fig:algo}.

Let $A_\mu = U_0\Lambda U_0^*$ be a spectral decomposition of $A_\mu$, where $U_0$ is a unitary matrix and $\Lambda$ is an $n\times n$ diagonal matrix containing the eigenvalues of $A_\mu$. 
We first remove from $U_0$ all columns corresponding to eigenvalue zero; we denote the new matrix by $U_0$ again. 
The remaining columns of $U_0$ span $\ker(A_\mu)^\perp$, the orthogonal complement to the kernel of $A_\mu$.

\begin{figure}[ht]
\begin{center}
\fbox{
\begin{minipage}{0.73\textwidth}
\begin{enumerate}
        \item Compute $A_\mu=[\sqrt{\mu_i\mu_j}a_{i,j}]_{i,j}$.
    \item Find the spectral decomposition $A_\mu=U_0\Lambda U_0^*$.
    \item Remove all columns from $U_0$ that correspond to zero eigenvalues of $A_\mu$.
    \item Set $k_i=\mu_iN$. Compute $V=[v_{i,j}]_{N\times n}$ as
    $$v_{i,j}=\left\{\begin{array}{ll}
       \frac{1}{\sqrt{k_j}}  &  \mbox{ if }\  \sum_{s=1}^{j-1}k_s<i\leq \sum_{s=1}^{j}k_s\\
       0  & \mbox{ otherwise} 
    \end{array} \right.$$ 
    \item Compute $U=VU_0$.
\item The SBM-driven Fourier transform is defined as follows: 
\begin{enumerate}
    \item For each distinct eigenvalue $\lambda$, the Fourier projection is 
    $$\widehat{x}(\lambda)=\sum_{i:\ \lambda_i=\lambda}\langle  x, U_i\rangle U_i.$$
    \item  Compute $\widehat{x}(0)= x-\sum_{\lambda\not=0}\widehat{x}(\lambda). $
    \item The inverse Fourier transform is $x=\sum_\lambda \widehat{x}(\lambda)$.
\end{enumerate}
\end{enumerate}
\end{minipage}
}
\end{center}
\caption{Algorithm to compute the SBM-driven Fourier transform of a signal $x$ on a graph sampled from $\SBM(A,\mu,N)$.}
\label{fig:algo}
\vspace{-0.4cm}
\end{figure}

If all nonzero eigenvalues of $A_\mu$ are distinct then we can adopt a simplified approach. In this case, all eigenspaces corresponding to nonzero eigenvalues are 1-dimensional, so we can represent the Fourier coefficients as scalars.
Let $U=VU_0$, where $V$ is the isometry given in \cref{def:M-A-mu-D} \ref{def:V-isom}. By \cref{prop:lin-alg-compression},  the columns of $U$ are exactly the eigenvectors of $W$ corresponding to its nonzero eigenvalues, and they form an orthonormal basis for $\ker(W)^\perp$.  The SBM-driven Fourier transform of the  graph signal $x\in \bbC^N$ can then be computed as $\widehat x = U^*x$. In particular, the Fourier coefficient associated with {non-zero} eigenvalue $\lambda_i$ equals the inner product $\widehat{x}(\lambda_i)=\langle x,U_i\rangle$, where $U_i$ is the eigenvector of $W$ associated with eigenvalue $N\lambda_i$. The Fourier coefficient $\widehat{x}_0$ associated with eigenvalue 0 can be computed as $\widehat{x}_0=\|P_0(x)\|=\| x-\sum_i \langle x,U_i\rangle U_i\|$. The (partial) inverse transform can be obtained as $\widetilde{x}=U\widehat{x} $. If $x\in \ker(W)^\perp$, then $\widehat{x}_0=0$ and $\widetilde{x}=x$. Otherwise, $\widetilde{x}=x-P_0(x)$.

In the general case where $A_\mu$ has repeated eigenvalues,  we cannot use the individual inner products $\langle x,U_i\rangle$ to represent the Fourier coefficients since they do not have the required convergence properties (see \cref{sec:graphonFT}). 
In this case, we need to adopt the graphon Fourier transform described in \cref{def:graphonFT}. Adapted 
to graph signals in $\bbC^N$, the Fourier projection 
corresponding to nonzero eigenvalue $\lambda$ of $A_\mu$ is represented by the vector 
$$
\widehat{x} (\lambda)=\sum_{i:\lambda_i=\lambda} \langle x, U_i\rangle U_i.
$$
For eigenvalue zero, we have that $\widehat{x}(0)=x-\sum_{i:\lambda_i\neq 0}\widehat{x}(\lambda)$.
The inverse Fourier transform of $x$ is given as $x=\sum_\lambda \widehat{x}(\lambda)$, where the sum is over all eigenvalues of $A_\mu$, including zero. 
 


\section{Fourier basis for SBM with group symmetries}\label{sec:cayley-general}
In this section, we discuss stochastic block models whose structure is informed by certain group symmetries.
We formalize the group theoretic symmetries of an SBM through its probability matrix as follows.
\begin{definition}\label{def:CayleySBM-matrix}
An $n\times n$ matrix $A$ is called a \emph{Cayley matrix} on a group $\bbG$
if $|\bbG|=n$ and there exists a \emph{connection function} $f:\bbG\rightarrow [0,1]$ so that $a_{i,j}=f(g_i^{-1}g_j)$ for all $i,j\in [n]$. The function $f$ must be \emph{inverse-invariant}, i.e.\ $f(x)=f(x^{-1})$. An $\SBM(A,\mu,n)$ is called a \emph{$\bbG$-SBM}, if $A$ is a Cayley matrix on the group $\bbG$. 
\end{definition}

A symmetric matrix can be viewed as a weighted graph, where the edge between vertices $i$ and $j$ has weight $a_{ij}$. With this interpretation,  the definition of a Cayley matrix coincides with that of a Cayley graph.
This fact has inspired our choice of terminology in \cref{def:CayleySBM-matrix}.

\subsection{$\bbG$-SBM with uniform measure}\label{subsec:G-SBM-uniform}
Let $\SBM(A,\mu,N)$ be a $\bbG$-SBM with the connection function $f:\bbG\to [0,1]$. 
Here, we consider the special case where $\mu $ is the uniform measure on $[n]$. 
Under this condition,  $A_\mu =\frac{1}{n}A$, which is a Cayley matrix,
and the associated graphon $w_{A,\mu}$ becomes a \emph{Cayley graphon}. 
Cayley graphons are generalizations of Cayley graphs, and were first introduced in \cite{cayley-graphon}.
The Cayley graphons of finite groups are precisely the graphons associated with a $\bbG$-SBM with uniform measure.
For such SBM, the graphon Fourier basis can be derived directly from the Cayley matrix $A$.
Eigen-decompositions of Cayley matrices are well-understood when the group is Abelian (\cite{1979:Babai:Spectraofcayleygraphs}) or quasi-Abelian (\cite{Rockmore}). 
Representation theory of groups has been used to construct `preferred' Fourier bases  for signal processing on Cayley graphs and graphons (see e.g.~\cite{MDK:2019:Sampta,Ghanehari-Janssen-Kalyaniwalla,BGHP-2023,permutahedron}).

Using harmonic analysis, it is easy to find a graphon Fourier basis for a $\bbG$-SBM where $\bbG$ is Abelian.
Let $\bbG=\{g_1,\ldots,g_n\}$ be an Abelian group, and $\bbT$ be the multiplicative group of complex numbers with modulus 1 represented as $\bbT=\{e^{2\pi i x}:0\leq x<1\}$. 
Characters of $\bbG$ are maps $\chi:\bbG\to {\mathbb T}$ satisfying $\chi(gh)=\chi(g)\chi(h)$ for all $g,h\in {\mathbb G}$. 
We denote the collection of all characters of $\bbG$ by $\widehat{\bbG}$. The set of normalized characters $\left\{\frac{\chi}{\sqrt{n}}: \ \chi\in \widehat{\bbG}\right\}$ forms an orthonormal basis for the vector space $\bbC^n$. 
Here, we think of $\bbC^n$ as vectors indexed by elements of $\bbG$. This allows us to identify functions on $\bbG$ with vectors in $\bbC^n$.

The characters of a group give a  diagonalization of a Cayley matrix on $\bbG$ as follows. For a proof, see e.g., \cite{1979:Babai:Spectraofcayleygraphs,BGHP-2023}.
\begin{lemma}[Diagonalization of Cayley matrices]\label{lem:diagonalization-Cayley} 
Let $\mathbb{G}$ be an Abelian group,  and $f:\bbG\rightarrow [0,1]$ be a connection function. 
Let $A$ be the Cayley matrix on $\bbG$ defined by $f$. Let $\widehat{\bbG}=\{\chi_1,\ldots,\chi_n\}$ be the set of characters of $\bbG$. Then every character $\chi_i\in\widehat{\bbG}$ is an eigenvector of $A$ associated with eigenvalue $\lambda_i=\Sigma_{x\in \bbG}f(x)\overline{\chi_i(x)}$. 
Consequently, the unitary matrix $U$ defined as $U:=\left[\frac{\chi_j(g_i)}{\sqrt{n}}\right]_{i,j}$ diagonalizes $A$: 
$$
U^{*} A U =\Gamma:= {\rm diag}\left({\sum\limits_{x\in \bbG} f(x)\overline{\chi_1(x)}},\ldots,\sum\limits_{x\in \bbG}f(x)\overline{\chi_n(x)}\right).
$$
\end{lemma}
The characters of $\bbG$ thus provide an eigenbasis for a Cayley matrix on a group $\bbG$ no matter which connection function $f$ is used.

\begin{remark}\label{remark:real-eigen}
For $\chi\in \widehat{\bbG}$, let $\overline{\chi}$ denote the character defined by $\overline{\chi}(x)=\overline{\chi (x)}$  for all $x\in \bbG$. Since $\overline{\chi (x)}=\chi(x^{-1})$ and $f$ is inverse-invariant, we have that the eigenvalues associated with $\chi$ and $\overline{\chi}$ are equal. Therefore, $A$ has a repeated eigenvalue for each character $\chi$ so that $\chi\not=\overline{\chi}$. Moreover, $\chi+\overline{\chi} $ and {$i(\chi-\overline{\chi})$} provide a real-valued pair of eigenvectors that span the same space as $\chi$ and $\overline{\chi}$.
\end{remark}
\begin{theorem}[GFT for samples of $\bbG$-SBM with uniform measure]
\label{prop:GSP-Cayley-uniform}
Let $W$ be the model matrix of $\SBM(A,\mu,N)$, where $A$ is a Cayley matrix on an Abelian group $\bbG$ with the connection function $f$. Suppose $\mu$ is the uniform measure on $[n]$. 
Then  $\{\frac{N}{n}\sum\limits_{x\in \bbG} f(x)\overline{\chi(x)}\}_{\chi\in\widehat{\bbG}}$ contains all nonzero eigenvalues of $W$ with corresponding orthonormal eigenbasis $\{\frac{1}{\sqrt{N}}D \chi\}_{\chi\in\widehat{\bbG}}$,
where $D$ is given in \cref{def:M-A-mu-D}.
\end{theorem}
\begin{proof}
    Since $\mu$ is a uniform measure, we have $A_\mu=\frac{1}{n}A$, $M=\frac{1}{n}I$, and $V=\sqrt{\frac{n}{N}}D$. By \cref{lem:diagonalization-Cayley},  the collection $\{\frac{1}{n}\sum\limits_{x\in \bbG} f(x)\overline{\chi(x)}\}_{\chi\in\widehat{\bbG}}$ includes all eigenvalues of $A_\mu$ with corresponding orthonormal eigenbasis $\{ \frac{1}{\sqrt{n}}\chi\}_{\chi\in\widehat{\bbG}}$. The statement then follows directly from the relation between the eigen-decompositions of $A_\mu$ and $W$ given in \cref{prop:lin-alg-compression}.
\end{proof}

In this paper, we restrict our attention to Cayley graphons on Abelian groups. 
If the underlying group is not Abelian, the representation theory of the group may be used to develop a specific basis for the graphon Fourier transform. 
This basis can then be used to provide an SBM-driven Fourier transform for graphs sampled from the SBM.
Details can be found in \cite{Ghanehari-Janssen-Kalyaniwalla}; for applications on the symmetric group, see \cite{BG-EUSIPCO,permutahedron}.
\subsection{General robustness results under block size perturbation}
\label{subsec:perturbations}
In this section, we explore how the SBM-driven Fourier transform
defined on samples of an SBM varies when the block ratios in the SBM are adjusted slightly. 
We remark that the results in this subsection are valid for general SBMs. We apply these general results to the Cayley setting in the subsequent sections. 

As shown in \cref{subsec:computing-basis},
the eigen-decomposition of the $n\times n$ weighted probability matrix $A_\mu$ leads to a basis for the graphon Fourier transform for the SBM defined by $A$ and $\mu$. 
This point of view allows us to use tools from matrix perturbation theory to control the effect of small variations in block sizes on the resulting graphon Fourier transform.

Let $A$ be a probability matrix of size $n$. Fix $N\in {\mathbb N}$ and $\epsilon >0$. Let $\mu,\mu'$ be probability measures on $[n]$ so that  $\mu'_i=\mu_i(1+\epsilon_i)$ and $|\epsilon_i|\leq \epsilon$ for all $i\in [n]$. Let $W,W'$ be the model matrices of $\SBM(A,\mu,N)$ and $\SBM(A,\mu',N)$, respectively. Let
$\{y_i\}_{i\in I_{A,\mu}}$ and $\{y_i'\}_{i\in I_{A,\mu'}}$  be orthonormal sets of eigenvectors associated with nonzero eigenvalues of $W$ and $W'$ ordered as in \eqref{eq:ordering}, respectively. 

The following theorem presents an upper bound on the change to $\widehat{x}$ if we use the basis  $\{y'_i\}$ derived from the slightly different measure $\mu'$, instead of the true basis  $\{y_i\}$. First, for any eigenvalue $\lambda$ of $A_\mu$ we define $\gamma(\lambda)$ as the gap between $\lambda $ and the  other eigenvalues, i.e.,
$$
\gamma (\lambda)=\min\{ |\lambda -\lambda_i|:\lambda_i\not=\lambda\}.
$$

\begin{theorem}
\label{prop:perturbed_blocksizes}
Consider any  nonzero eigenvalue $\lambda$  of $A_\mu$, with multiplicity $d$.
Then for any signal $x\in \bbC^N$,
\begin{equation}\label{eq:4-1}
    \| \widehat{x}(\lambda)-\sum_{i\in I_\lambda} \langle x,y'_i\rangle_{{\bbC}^N} y'_i \|_{\bbC^N} 
    \leq  \left(\frac{2^{5/2}d^{1/2} \|A_\mu\|_{\opr}}{\gamma(\lambda)}
    n\epsilon+\frac{2\sqrt{3}}{\sqrt{\mu_{\min}}}\sqrt{n\epsilon}\right) \| x\|_{\bbC^N},
\end{equation}
where $\mu_{\min}$ denotes the minimum value attained by the measure on $[n]$, and $d$ is the multiplicity of $\lambda$.
\end{theorem}
The proof of \cref{prop:perturbed_blocksizes} depends heavily on the Davis--Kahan theorem on matrix perturbations. The relevant background can be found in \cref{sec:perturb}. First, we prove a necessary lemma that shows how the isometry $V$ is affected by the change in measure.

\begin{lemma}\label{lem-V-perturbation}
%
Let $V:\bbC^n\rightarrow \bbC^N$ (resp.~$V':\bbC^n\rightarrow \bbC^N$) be the isometry given in \cref{def:M-A-mu-D} for $\mu$ (resp.~${\mu'}$). Then, letting $\mu_{\min}:=\min\{\mu_1,\ldots,\mu_n\}$, we have
$$\|V-V'\|_\opr\leq \sqrt{\frac{3n}{\mu_{\min}}}\sqrt{\epsilon}.$$
\end{lemma}
\begin{proof}
{Let $D,M$ (resp.~$D',M'$) be the matrices associated with $\mu$ (resp.~$\mu'$) as in \cref{def:M-A-mu-D}. 
Recall that $k_i=\mu_i N$ and $k_i'=\mu'_iN$.
Using the structure of $D$ and $D'$, it is easy to observe that the $j$'th column of $D-D'$ has nonzero elements in at most 
$|\sum_{i=1}^{j-1} k_i-\sum_{i=1}^{j-1} k'_i|+|\sum_{i=1}^{j} k_i-\sum_{i=1}^{j} k'_i|$ entries. Given that the matrix $D-D'$ has $n$ columns, and its nonzero entries are either $1$ or $-1$, we have 
\begin{eqnarray}
\|D-D'\|_\opr&\leq& \|D-D'\|_F\leq \sqrt{n\left(\left|\sum_{i=1}^{j-1} k_i-\sum_{i=1}^{j-1} k'_i\right|+\left|\sum_{i=1}^{j} k_i-\sum_{i=1}^{j} k'_i\right|\right)}\nonumber\\
&\leq& \sqrt{2n\sum_{i=1}^{n} |k_i-k'_i|}=\sqrt{2n\sum_{i=1}^{n} |\epsilon_i|\mu_i N}\leq  \sqrt{2n\epsilon N\sum_{i=1}^{n} \mu_i } = \sqrt{2 \epsilon nN}, \label{eq-D-D'-stuff}
\end{eqnarray}
where in the last equality we used the fact that  $\sum_{i=1}^n\mu_i= 1$.
Next, note that $\|M^{-\frac{1}{2}}\|_\opr=\frac{1}{\sqrt{\mu_{\min}}}$ and 
\begin{eqnarray*}
\|M^{-\frac{1}{2}}-M'^{-\frac{1}{2}}\|_\opr&=&\max\left\{ \left|\frac{1}{\sqrt{\mu_i}}-\frac{1}{\sqrt{\mu_i(1+\epsilon_i)}}\right|:i\in[n]\right\}\\
&=&\max\left\{\frac{1}{\sqrt{\mu_i}}\left|\frac{1-\sqrt{1+\epsilon_i}}{\sqrt{1+\epsilon_i}}\right|: i\in [n]\right\}
\leq \frac{\epsilon}{2\sqrt{\mu_{\min}}}.
    \end{eqnarray*}
These facts, together with \eqref{eq-D-D'-stuff} and \cref{lem:relations-between-matrices} \cref{lem-part:D-norm}, lead to the desired estimates:
\begin{eqnarray*}
    \|V-V'\|_\opr&=&\frac{1}{\sqrt{N}}\|DM^{-\frac{1}{2}}-D'M'^{-\frac{1}{2}}\|_\opr\\
    &\leq& \frac{1}{\sqrt{N}}\left(\|D-D'\|_\opr\|M^{-\frac{1}{2}}\|_\opr+\|D'\|_\opr\|M^{-\frac{1}{2}}-M'^{-\frac{1}{2}}\|_\opr\right)\\
    &\leq& \frac{1}{\sqrt{N}}\left(\sqrt{2\epsilon nN}\frac{1}{\sqrt{\mu_{\min}}}+\sqrt{N} \frac{\epsilon}{2\sqrt{\mu_{\min}}}
    \right)\leq \frac{\sqrt{3n}}{\sqrt{\mu_{\min}}}\sqrt{\epsilon}.
\end{eqnarray*}   
} 
\end{proof}

\begin{proof}[Proof of \cref{prop:perturbed_blocksizes}]
Since $\mu'$ is a probability measure on $[n]$, we have  that $\sum_{i=1}^n \epsilon_i=0$. 
Note that
$$A_{\mu'}-A_{\mu}=\left[\sqrt{\mu_i}\sqrt{\mu_j}{a_{ij}}(\sqrt{1+\epsilon_i}\sqrt{1+\epsilon_j}-1)\right]_{i,j}=A_\mu\circ E,$$ 
where
$E=\left[\sqrt{1+\epsilon_i}\sqrt{1+\epsilon_j}-1\right]_{i,j}$ and $\circ$ denotes the Schur (or Hadamard) product of matrices. Let $z, {\mathbf 1}\in \bbC^n$ be defined as 
$z=[\sqrt{1+\epsilon_1},\ldots, \sqrt{1+\epsilon_n}]^\t$, ${\mathbf 1}=[1,\ldots, 1]^\t$, and note that 
$E=(z-\mathbf{1})z^\t+\mathbf{1}(z^\t-\mathbf{1}^\t)$.
Moreover, observe that $\|z\|_{\bbC^n}\leq \sqrt{n(1+\epsilon)}$ and $\|z-\mathbf{1}\|_{\bbC^n}\leq \sqrt{n}(\sqrt{1+\epsilon}-1)\leq\sqrt{n}(\frac{\epsilon}{2})$.
Since the operator norm is submultiplicative with respect to the Schur product of matrices, using the above norm estimates we get 
\begin{eqnarray*}
    \|A_{\mu'}-A_{\mu}\|_\opr
    &\leq& \|A_{\mu}\|_\opr\|E\|_\opr\\ &\leq & \|A_{\mu}\|_\opr(\|z-\mathbf{1}\|_{\bbC^n}\|z\|_{\bbC^n}+\|\mathbf{1}\|_{\bbC^n}\|z-\mathbf{1}\|_{\bbC^n})\\
    &\leq& n(\frac{\epsilon}{2})(\sqrt{1+\epsilon}+1)\|A_{\mu}\|_\opr.
\end{eqnarray*}
Since $\epsilon \leq 1$, we have the upper bound $\|A_{\mu'}-A_{\mu}\|_\opr\leq 2\epsilon n\|A_{\mu}\|_\opr$.

Let $V:\bbC^n\rightarrow \bbC^N$ (resp.~$V':\bbC^n\rightarrow \bbC^N$) be the isometry given in \cref{def:M-A-mu-D} for $\mu$ (resp.~${\mu'}$), and 
set $x_i=V^\t y_i$ and $x'_i={V'}^\t y'_i$. 
By \cref{prop:lin-alg-compression},   $\{x_i\}_{i\in I_{A,\mu}}$ (resp.~$\{ x'_i\}_{i\in I_{A,\mu'}}$) is an orthonormal set of eigenvectors of $A_\mu$ (resp.~$A_{\mu'}$) associated with nonzero eigenvalues, indexed as in \eqref{eq:ordering}.
Since all matrices $A_\mu$, $W$, $A_{\mu'}$ and $W'$ are symmetric and real-valued, their corresponding eigenbases can be chosen to be real-valued vectors. 
Fix a nonzero eigenvalue $\lambda $ of $A_\mu$ of multiplicity $d$ so that $\lambda=\lambda_r=\dots =\lambda_s$, where $d=s-r+1$. Let  $\ee ={\rm span}\{ x_r,\dots ,x_s\}$ and  $\ee' ={\rm span}\{ x'_r,\dots ,x'_s\}$. Let $P_\ee$ and $P_{\ee'}$  be the orthogonal projections, in $\bbR^n$, on $\ee$ and $\ee'$, respectively.  
Then, by \cref{DavisKahanProjections}, we have that
\begin{equation}\label{eqn:Pee}
   \| P_\ee-P_{{\ee}'}\|_{\opr}\leq \| P_\ee-P_{{\ee}'}\|_{F}\leq \frac{2^{3/2}d^{1/2}\|A_\mu-A_{\mu'}\|_{\opr}}{\min\{ \lambda_{r-1}-\lambda,\lambda-\lambda_{s+1}\}}
   \leq  \frac{2^{5/2}d^{1/2}\epsilon n\|A_\mu\|_\opr}{\gamma(\lambda)}.
 \end{equation}
Note that the projection $P_\ee:\bbC^n\rightarrow\bbC^n$ is defined as 
$$P_\ee(z)=\sum_{i=r}^s \langle z,x_i\rangle x_i =\sum_{i=r}^s (x_i^*z)x_i=\sum_{i=r}^s x_i(x_i^*z)=(\sum_{i=r}^s x_ix_i^*)z.$$
So, when represented in matrix form, we have that $P_\ee=\sum_{i=r}^s x_ix_i^*$; a similar formula holds for $P_{\ee'}$. 
Let $P_E$ and $P_{E'}$ be the orthogonal projections on $E={\rm span}\{ y_r,\dots , y_s\}$ and $E'={\rm span}\{ y'_r,\dots , y'_s\}$ respectively. By \cref{prop:lin-alg-compression} \ref{part-prop:lin-alg:Y-eigen->X-eigen}, each $y_i$ (resp.~$y_i'$) is in the image of $V$ (resp.~$V'$), so $Vx_i=VV^\t y_i=y_i$ and similarly $V'x'_i=y'_i$.
Expressing $P_E$ in matrix form, we have that 
$$P_{E}=\sum_{i=r}^s (Vx_i)(Vx_i)^*
=\sum_{i=r}^s Vx_i\, x_i^*V^*
=V\left(\sum_{i=r}^s x_i\, x_i^*\right)V^*=VP_\ee V^*.
$$
Applying the same argument, we obtain that $P_{E'}=V'P_{\ee '}{V'}^*$.
So 
\begin{equation*}
P_E-P_{E'}=VP_\ee V^*-V'P_{\ee '}{V'}^*
=V(P_\ee -P_{\ee '})V^*+(V-V')P_{\ee '}V^*+V'P_{\ee '}(V^*-{V'}^*),
\end{equation*}
which implies that 
\begin{equation*}
\|P_E-P_{E'}\|_{\opr}
\leq
\|P_\ee -P_{\ee '}\|_{\opr}+\|V-V'\|_{\opr}\|P_{\ee '}\|_{\opr}+\|P_{\ee '}\|_{\opr}\|V^*-{V'}^*\|_{\opr},
\end{equation*}
where we used the facts that operator norm is submultiplicative, and $\| V\|_{\opr}=\| V^*\|_{\opr}=1$, as $V$ is an isometry.
Furthermore, since any projection is contractive, we get 
\begin{equation}\label{eq:to-bound}
\|P_E-P_{E'}\|_{\opr}
\leq
\|P_\ee -P_{\ee '}\|_{\opr}+2\|V-V'\|_{\opr}.
\end{equation}
Combining \eqref{eq:to-bound}, \eqref{eqn:Pee} and \cref{lem-V-perturbation}, we obtain that 
$$\| P_E-P_{E'}\|_\opr\leq \frac{2^{5/2}d^{1/2}\epsilon n\|A_\mu\|_{\opr}}{\gamma(\lambda)}
+2\sqrt{\frac{3n\epsilon}{\mu_{\min}}},$$
which finishes the proof.
\end{proof}

\begin{remark}[Sharpness of the upper bound in \Cref{prop:perturbed_blocksizes}]
The values $n,\gamma(\lambda), \mu_{\min}, d$ and $\|A_{\mu}\|_\opr$ are determined by $A$ and $\mu$. 
So, these quantities can be regarded as constants, since $A$ and $\mu$ are fixed. 
    The constants $n,\gamma(\lambda), d$ and $\|A_\mu\|_\opr$, together with the error term $\epsilon$, appear in \eqref{eq:4-1} as a consequence of applying Davis--Kahan theorem, which is a standard and efficient tool for bounding changes in 
    eigenspaces of a symmetric matrix under additive perturbations.
    To show that the bound in \eqref{eq:4-1} cannot be improved, we observe that the term involving $\sqrt{\epsilon}$ cannot be replaced with a constant multiple of $\epsilon.$
    Indeed, it is not difficult to see that $\|V-V'\|_\opr\gtrsim\sqrt{\epsilon}$. 
    Finally, note that $N$ can grow large, whereas $n$, which denotes the number of blocks in the SBM, remains fixed. 
    Therefore, it is important that the right hand side of \eqref{eq:4-1} does not depend on $N$.
\end{remark}

\cref{prop:perturbed_blocksizes} shows that the difference in the Fourier transform is small if $\epsilon$, the bound on the relative perturbation of the block sizes, is much smaller  than both $\frac{\mu_{\min}}{n}$ and $\frac{\gamma(\lambda)}{nd^{1/2}\|A_\mu\|_{\opr}}$. 
In particular, in the setting of a $\bbG$-SBM with almost uniform block sizes, \cref{prop:perturbed_blocksizes} gives the conditions under which the well-structured group Fourier basis obtained in \Cref{prop:GSP-Cayley-uniform} can be used as a good approximation for the $\bbG$-SBM Fourier basis. 
In the next section, we study $\bbG$-SBM with block sizes that are not close to uniform.
\subsection{$\bbG$-SBM with general measure}\label{subsec:Cayley-general}
For samples of a $\bbG$-SBM with non-uniform measure, the character basis  described in \cref{subsec:G-SBM-uniform} is not necessarily a good choice.
Namely, this basis does not provide convergence of the graph Fourier transform as described in \cref{thm:convergence}. 
In this section we discuss how the group symmetries of the Cayley matrix help to partially define a graphon Fourier basis
(\cref{thm: general-decomp}).
First, we introduce the necessary notation.
\begin{notation}\label{notation:index-identofoer}
We define the index identifier function $\iota$ as follows.
 Fix the ordering $\widehat{\bbG}=\{\chi_1,\ldots,\chi_n\}$, and assume that $\chi_1$ is the trivial representation (i.e., $\chi_1\equiv 1$).   
For every $\chi\in\widehat{\bbG}$, let $\iota(\chi)$ denote the index $i$ such that $\chi=\chi_i$.
\end{notation}

The next theorem reduces the problem of computing an eigenbasis for the model matrix $W$ to the same problem for the smaller $n\times n$ matrix 
$\widetilde{M}\Gamma$. Both $\widetilde{M}$ and $\Gamma$ are  derived from the characters of $A$ and are easy to compute. This reduction may give us insight into properties of the eigenvectors of $W$ that might not be obvious from $A_\mu$ (e.g., see \cref{prop:one-large-block}).

\begin{theorem}[GFT for samples of $\bbG$-SBM with general measure]
\label{thm: general-decomp}
Let $W$ be the model matrix of $\SBM(A,\mu,N)$, where $A$ is a Cayley matrix on an Abelian group $\bbG$.
Define 
$$\widetilde{M}=\left[\frac{1}{n}\sum_{j=1}^n \mu_j(\chi_k^{-1}\chi_\ell)(g_j)\right]_{k,\ell\in[n]}.$$
Let $U,\Gamma$ be so that $A=U\Gamma U^*$ is the diagonalization of $A$ given in \cref{lem:diagonalization-Cayley}, 
and let $D$ be as in \cref{def:M-A-mu-D}. 
For a unit vector $y\in\bbC^N$ and $\lambda\neq 0$, we have
 \begin{equation*}
    y \ \mbox{ is a }\   \lambda\mbox{-eigenvector of } W \iff \ \frac{1}{\sqrt{N}}(U^*D^\t)y \ \mbox{ is a } \frac{\lambda}{N}\mbox{-eigenvector of } \ \widetilde{M}\Gamma.
\end{equation*}
Moreover,  $\widetilde{M}$ satisfies certain symmetries, in the sense that  
$\widetilde{M}_{k,\ell}=\widetilde{M}_{1,\iota(\chi_k^{-1}\chi_\ell)}$ 
for $\chi_k,\chi_\ell\in\widehat{\bbG}$.
\end{theorem}
\begin{proof}
Let $y\in\bbC^N$, $z\in\bbC^n$, and $\lambda\neq 0$. With notation as in \cref{def:M-A-mu-D}, we have
\begin{equation}\label{part-thm: general-decomp-Y-Z-rel} 
y = \frac{1}{\sqrt{N}}(DM^{-1}U)z \iff z= \frac{1}{\sqrt{N}}(U^*D^\t)y \ \mbox{ and }\ y\in{\rm range}(D).
\end{equation}
In fact, \cref{part-thm: general-decomp-Y-Z-rel} can be easily verified using the equation $\frac{1}{N}D^\t D =M$ (\cref{lem:relations-between-matrices} \cref{part-lem:W=DADt-and DtD=NM}). 
Moreover, if $y,z$ satisfy  the relation described in \eqref{part-thm: general-decomp-Y-Z-rel}, then $y$ is nonzero precisely when $z$ is nonzero. This follows from the fact that $DM^{-1}U$ and $U^*D^\t D=N(U^*M)$ have null kernels. 

Next note that putting \cref{part-prop:lin-alg:Y-eigen->X-eigen} and \cref{part-prop:lin-alg:X-eigen->Y-eigen} of \cref{prop:lin-alg-compression} together, for $\lambda \neq 0$, we have the following equivalence:
        $$Wy = \lambda y \quad \mbox{ \emph{iff} } \quad \exists x\in\bbC^n,\ y= Vx \text{ and }A_{\mu}x = \frac{\lambda}{N}x.$$     
Substituting $A=U\Gamma U^*$ in the expression $A_\mu=M^{\frac{1}{2}}AM^{\frac{1}{2}}$, the equality
$A_{\mu}x = \frac{\lambda}{N}x$ can be equivalently written as 
$\Gamma U^{*}M^{\frac{1}{2}}x=\frac{\lambda}{N}U^*M^{-\frac{1}{2}}x.$ Letting $z=U^*M^{\frac{1}{2}}x$ (equivalently $x=M^{-\frac{1}{2}}Uz$), the above \emph{iff} statement can be rewritten as 
$$Wy = \lambda y \quad \mbox{ \emph{iff} } \quad \exists z\in\bbC^n,\ y = VM^{-\frac{1}{2}}Uz \text{ and }
\Gamma z = \frac{\lambda}{N}(U^*M^{-1}U)z.$$ 
From $\Gamma z = \frac{\lambda}{N}(U^*M^{-1}U)z$, we observe that $z$ is a $\lambda$-eigenvector of $N(U^*MU)\Gamma$. 
In addition, 
$ y = Vx=\frac{1}{\sqrt{N}}(DM^{-1}U)z$, so $z$ and $y$ satisfy the equations from \cref{part-thm: general-decomp-Y-Z-rel}. 
So to finish the proof, we only need to show that $\widetilde{M}=U^* M U$.

For $k,l\in[n]$, a direct calculation shows that the $(k, l)$-entry of $U^*MU$ is given by 
$$\frac{1}{n}\sum_{j=1}^n \mu_j\overline{\chi_k(g_j)}\chi_l(g_j)=\frac{1}{n}\sum_{j=1}^n \mu_j(\chi_k^{-1}\chi_l)(g_j).$$
Finally, let $\chi,\chi_k,\chi_\ell\in\widehat{\bbG}$ and suppose $\iota(\chi\chi_k)=k',\iota(\chi\chi_\ell)=\ell'$. Then, the 
$(k', \ell')$-entry of $U^*M U$ is given by 
$\frac{1}{n}\sum_{j=1}^n \mu_j\overline{\chi_{k'}(g_j)}\chi_{\ell'}(g_j)=\frac{1}{n}\sum_{j=1}^n \mu_j\overline{\chi(g_j)\chi_{k}(g_j)}\chi(g_j)\chi_{\ell}(g_j),$
    which simplifies to the $(k, \ell)$-entry of $U^*M U$, since $\overline{\chi(g_j)}\chi(g_j)=1$ for all $j$. Taking $\chi=\chi_k^{-1}$, and noting that $\chi_1$ is the identity element of the group $\widehat{\bbG}$ finishes the proof. 
\end{proof}

The applicability of \cref{thm: general-decomp} in practice is due to the fact that $\widetilde{M}$ has certain symmetric features, which result in  
simplified computations for the eigen-decomposition of $\widetilde{M}\Gamma$. We discuss the case of $\bbG$-SBM models with a dominant block below, where we can offer a recipe for computing eigenvectors of $\widetilde{M}\Gamma$ with zero mean.
\begin{proposition}\label{prop:one-large-block}
Let $\bbG$ be an Abelian group of size $n$ with neutral element $e_{\bbG}$.
Let $W$ be the model matrix of $\SBM(A,\mu,N)$, where $A$ is a Cayley matrix on $\bbG$, and let $V$ be the isometry as in \cref{def:M-A-mu-D}.
Let $\tau \in (0,\frac{1}{n})$ be such that $\tau N$ is an integer. Define the measure $\mu$ as follows:
$$\mu(\{e_{\bbG}\})=1-\tau(n-1)\ \mbox{ and }\ \mu(\{g\})=\tau, \forall \ g\in \bbG\setminus\{e_{\bbG}\}.$$
\begin{itemize}
\item[(i)] Let $m>1$. If $\gamma$ is an eigenvalue of $A$ of  multiplicity $m$ then
$N\tau \gamma$ is an eigenvalue of $W$ of multiplicity at least $m-1$. 
\item[(ii)] Let $\{\chi_{\alpha_i}: 1\leq i\leq m\}$ be an eigenbasis of characters for the $\gamma$-eigenspace of $A$ as in \cref{lem:diagonalization-Cayley}. Then 
{$\left\{V\left((i-1)\chi_{\alpha_i}-\sum_{j=1}^{i-1}\chi_{\alpha_j}\right): 1<i\leq m\right\}$} are {orthogonal} $N\tau \gamma$-eigenvectors of $W$. 
\end{itemize}
\end{proposition}
\begin{proof}
Let $\bbG=\{g_1,\ldots,g_n\}$ be labeled such that $g_1=e_{\bbG}$.
Let $\chi_1$ denote the trivial character of $\bbG$, i.e., $\chi_1(g)=1$ for all $g\in\bbG$.
Using the orthogonality of characters,  we get
\begin{eqnarray*}
    \sum_{g\in \bbG\setminus\{e_{\bbG}\}}\chi_k^{-1}(g)\chi_l(g)&=&\sum_{g\in \bbG\setminus\{e_{\bbG}\}}\overline{\chi_k(g)}\chi_l(g)=\langle\chi_l,\chi_k\rangle-\overline{\chi_k(e_{\bbG})}\chi_l(e_{\bbG})=\left\{\begin{array}{cc}
        -1 & k\neq l \\
       n-1  & k=l
    \end{array}\right..
\end{eqnarray*}
So the matrix $\widetilde{M}$ from \cref{thm: general-decomp} is given by $\widetilde{M}=\frac{1-n\tau}{n}J+\tau I$, where $J$ is the all 1 matrix, and $I$ is the identity matrix. 

Now, suppose $A$ has a repeated eigenvalue $\gamma=\gamma_{\alpha_1}=\ldots=\gamma_{\alpha_m}$. So by \cref{lem:diagonalization-Cayley}, there exist characters $\chi_{\alpha_1},\ldots, \chi_{\alpha_m}$ in $\widehat{\bbG}$ such that
$\gamma=\sum_{g\in\bbG}f(g)\overline{\chi_{\alpha_1}(g)}=\ldots=\sum_{g\in\bbG}f(g)\overline{\chi_{\alpha_m}(g)},$
where $f$ is the connection function of $A$.
By definition,  the value $\gamma$ is repeated on the $\alpha_1,\ldots, \alpha_m$ entries of the diagonal matrix $\Gamma$. 
Let $\{e_i\}_{i=1}^{n}$ denote the standard basis of $\bbC^n$. 
The set {$\{\nu_i:=(i-1)e_{\alpha_i}-\sum_{j=1}^{i-1}e_{\alpha_j}:\ 1<i\leq m\}$ is an orthogonal} set of $\gamma$-eigenvectors of $\Gamma$ satisfying
{$J\nu_i=0$} for $2\leq i\leq m$. So, we have
$$
\widetilde{M}\Gamma {\nu_i}=\gamma (\frac{1-n\tau}{n}J+\tau I){\nu_i}=\tau\gamma {\nu_i}, \ \forall\ 2\leq i\leq m.
$$ 
Using the notations from \cref{thm: general-decomp}, let $y_i=\frac{1}{\sqrt{N}}DM^{-1}U{\nu_i}$ for $1<i\leq m$.
We have 
 \begin{eqnarray*}
     y_i&=&\frac{1}{\sqrt{N}}DM^{-1}{\left( (i-1)\chi_{\alpha_i}-\sum_{j=1}^{i-1}\chi_{\alpha_j}\right)}
    =VM^{-\frac{1}{2}}{\left( (i-1)\chi_{\alpha_i}-\sum_{j=1}^{i-1}\chi_{\alpha_j}\right)}\\
     &=&V\left(\frac{1}{\sqrt{\tau}}\left(I+{\rm diag}(\sqrt{\frac{\tau}{1-(n-1)\tau}}-1,0,\ldots,0)\right)\right)
     {\left( (i-1)\chi_{\alpha_i}-\sum_{j=1}^{i-1}\chi_{\alpha_j}\right)}.
 \end{eqnarray*}
 Note that the first coordinate of  ${U\nu_i}$ is equal to {$(i-1)\chi_{\alpha_i}(e_{\bbG})-\sum_{j=1}^{i-1}\chi_{\alpha_{j}}(e_{\bbG})=0$}, so ${\rm diag}(\sqrt{\frac{\tau}{1-(n-1)\tau}}-1,0,\ldots,0)({U\nu_i})=0$.
 So we get
 $y_i =\frac{1}{\sqrt{\tau}}V{\left( (i-1)\chi_{\alpha_i}-\sum_{j=1}^{i-1}\chi_{\alpha_j}\right)}.$
 The orthogonality of $\{y_i:\ 1<i\leq m\}$ follows from the orthogonality of $\{\nu_i:\ 1<i\leq m\}$ together with the fact that $U$ and $V$ are isometries.
\end{proof}

\begin{remark}[Alternative interpretation of \cref{thm: general-decomp}]
\label{remark-waighted}
The previous theorem can be understood in the context of weighted Fourier analysis on Abelian groups \cite{Benedetto92}. 
{Namely, we can interpret the effect of the weight $\mu$ as a change of measure on $\widehat{\bbG}$. 
Consider $\mu$ as the natural probability measure on $\widehat{\bbG}$ given by $\mu(\{\chi_i\})=\mu_i$, and define the weighted inner product space $\ell^2(\widehat{\bbG},\mu)$ via $\langle f,g\rangle_{\ell^2(\widehat{\bbG},\mu)}=\sum_{\chi\in\widehat{\bbG}} f(\chi)\overline{g(\chi)}\mu(\chi)$.
We have
\begin{equation}\label{eq-matrix-cal}
(\widetilde{M}\Gamma)_{k,l}=\frac{1}{n}\, (\mathcal{F} f)(\chi_l)\langle\chi_{l},\chi_{k}\rangle_{\ell^2(\widehat{\bbG},\mu)},\quad
\end{equation}
where $\mathcal{F} f$ denotes the classical Fourier transform of $f$ as a function on the Abelian group $\bbG$.
To interpret the above equation, note that when the block sizes are equal, the measure $\mu$ is uniform. In this case, the orthogonality of group characters imply that the matrix $\widetilde{M}\Gamma$ is diagonal (as it should be), since $\langle\chi_{l},\chi_{k}\rangle_{\ell^2(\widehat{\bbG},\mu)}=\frac{1}{n}\langle\chi_{l},\chi_{k}\rangle_{\bbC^n}=0$ when $l\neq k$. The formula in 
\eqref{eq-matrix-cal} allows us to witness non-uniformity of block sizes in terms of a non-diagonal $\widetilde{M}\Gamma$.
}
\end{remark}

\section{Example:  $\bbZ_5$-SBMs with different block sizes}
\label{sec:experiment}
To illustrate the theory presented in \cref{sec:main} and \cref{sec:cayley-general}, we consider stochastic block models based on the cyclic group $\mathbb{Z}_5$ with various block sizes. The group $\mathbb{Z}_5$ is equipped with addition modulo 5, and its elements are conventionally represented as \(\{0,1,2,3,4\}\). However, for consistency with the rest of the paper, we relabel these elements as \( g_i = i - 1 \) for \( 1 \leq i \leq 5 \).

We fix $A$ to be the Cayley matrix on $\bbZ_5$ associated with the connection function $f$ defined as
 $f(1)=f(4)=0.8$ and $f(i)=0.2$ for $i=0,2,3$. 
More precisely, the probability matrix $A=[a_{i,j}]$ is given by $a_{i,j}= f(j-i)$ for $1\leq i,j\leq 5$. 
Let $\mu=\{\mu_i\}_{i=1}^5$ be the measure representing the block sizes of SBMs.
All computations presented in this section were done  using Mathematica version 13.2.1. 
Throughout this section, we use notations from \cref{def:M-A-mu-D}.

The set of characters \(\widehat{\mathbb{Z}_5} = \{\chi_1, \dots, \chi_5\}\) is generated under pointwise multiplication by the character \(\chi: \mathbb{Z}_5 \to \mathbb{T}\), defined as \(\chi(j) = \omega^j\), where \(\omega = e^{\frac{2\pi i}{5}}\) is a primitive 5th root of unity. More explicitly, we set \(\chi_j = \chi^{j-1}\) for \(1 \leq j \leq 5\).  
Applying \cref{lem:diagonalization-Cayley}, the eigenvalues of \( A \) are given by  
\[
\lambda_i := \sum_{j=0}^{4} f(j) \omega^{-(i-1)j}, \quad \text{for } 1 \leq i \leq 5.
\]  
Since \( f(j) = f(-j) \), it follows that \(\lambda_2 = \lambda_5\) and \(\lambda_3 = \lambda_4\).  
As stated in \cref{thm: general-decomp}, we utilize the eigen-decomposition of \( \widetilde{M} \Gamma \) to determine the SBM Fourier basis for this graphon, where  $\Gamma={\rm diag}(\lambda_1,\ldots,\lambda_5)$ and 
\[
\widetilde{M} = \left[\frac{1}{5} \sum_{j=1}^{5} \mu_j (\chi_k^{-1} \chi_l)(g_j) \right]_{k,l} = \left[\frac{1}{5} \sum_{j=1}^{5} \mu_j \omega^{(l-k)(j-1)} \right]_{k,l}.
\]

\subsection{Comparing the graph Fourier basis with the SBM-driven Fourier basis}\label{sec:expA}
We first sampled from a $\bbZ_5$-SBM and compared the \emph{graph} Fourier basis with the SBM-driven Fourier basis. Recall that the graph Fourier basis is an eigenbasis of the graph adjacency matrix, while the SBM-driven Fourier basis is an eigenbasis of $W$, the model matrix. We sampled from $\SBM(A,\mu,N)$ where $A$ is the Cayley matrix described above, $N=6000$, and $\mu_1 = 1/3$, $\mu_2=\dots = \mu_5 = 1/6$ (so there is one block of size $2000$ and four of size $1000$).  We computed the $\bbZ_5$-SBM Fourier basis from the matrix $\widetilde{M}\Gamma$ as described in \cref{thm: general-decomp}. 

We took six different samples from the $\SBM(A,\mu,N)$.
The samples yielded very similar results, so we include the results of one of them.  \cref{tab:eigenvalues} gives the eigenvalues of the model matrix as given by \cref{thm: general-decomp} compared with the first six eigenvalues of the adjacency matrix of the sampled graph. 
As explained in \cref{subsec:GraphonModelMatrix}, since $\SBM(A,\mu,N)$ consists of five blocks, 
the SBM-driven Fourier basis only takes the eigenvectors corresponding to the first five eigenvalues into account, as the projection on the remaining eigenvectors may be viewed as noise. We include results for the sixth eigenvalue to show that indeed, $\lambda_6$ of the sample is significantly smaller (in magnitude) than the other eigenvalues. 
It is apparent from the table that the nonzero eigenvalues of $W$ have multiplicity 1, so  we have a unique eigenbasis for the range of $W$.

\begin{table}[ht]
\centering

\begin{tabular}{|l|c|c|c|c|c|c|}
\hline
 & $\lambda_1$ & $\lambda_2$ & $\lambda_3$ & $\lambda_4$ & $\lambda_5$ & $\lambda_6$ \\
\hline
Model matrix & 2622.1 & -1290.3 & -970.82 &468.1 & 370.8 & 0 \\
\hline
Sample adjacency matrix & 2622.2 & -1291.6 & -970.85 & 469.8 & 373.1 & -61 \\
\hline
\end{tabular}
    \caption{Eigenvalues of the model matrix $W$  compared with eigenvalues of the adjacency matrix of a sample.}
   \label{tab:eigenvalues}
    \vspace{-0.2cm}
\end{table}

We then compared the first five eigenvectors $\phi_i$ of the model matrix $W$ with the corresponding eigenvectors $\varphi_i$ of the adjacency matrix of the sampled graph by computing the absolute value of the inner product  $|\langle \varphi_i, \phi_i \rangle|$, to quantify alignment between the subspaces. The eigenvectors of $W$ were obtained by computing the eigenvectors of $\widetilde{M}\Gamma$ and applying \cref{part-thm: general-decomp-Y-Z-rel} to obtain $\phi_i = \frac{1}{\sqrt{N}}(DM^{-1}U)z_i$, where $z_1,\ldots, z_5$ are the eigenvectors of the matrix $\widetilde{M}\Gamma$ arranged according to eigenvalue magnitude in descending order. The close agreement in \cref{tab:W_and_sbm_evec_comparison} aligns with the convergence result stated in \cref{thm:convergence}.
\begin{table}[ht]
    \centering
    \begin{tabular}{|c|c|c|c|c|}
    \hline
    $|\langle \varphi_1, \phi_1\rangle |$ & $|\langle \varphi_2, \phi_2\rangle |$ & $|\langle \varphi_3, \phi_3\rangle |$ & $|\langle \varphi_4, \phi_4\rangle |$ & $|\langle \varphi_5, \phi_5\rangle |$ \\
    \hline
    0.99993 & 0.999712 & 0.999493 & 0.997782 & 0.996411 \\
    \hline
    \end{tabular}
    \caption{Absolute value of the inner product between the $i$-th eigenvector $\varphi_i$ of the sampled graph adjacency matrix and the corresponding eigenvector $\phi_i$ of the SBM model matrix $W$, computed via \cref{thm: general-decomp}.}
    \label{tab:W_and_sbm_evec_comparison}
\end{table}

\subsection{Comparing a transferred character basis with the SBM-basis}
\label{sec:VaryingBlockSize}
Let $A$ be the Cayley matrix as given in \cref{sec:expA}. Using the characters of $\bbZ_5$, we fix the eigenbasis of $A$ consisting of the following real-valued unit vectors: 
$$
\phi_1 = \frac{\chi_1}{\sqrt{5}},\, \phi_2 = \frac{\chi_{3}+\chi_{4}}{\sqrt{10}},\, \phi_3 = \frac{\chi_{4} - \chi_{3}}{i\sqrt{10}}, \, \phi_4 = \frac{\chi_{2}+\chi_{5}}{\sqrt{10}},\, \phi_5 = \frac{\chi_{2} - \chi_{5}}{i\sqrt{10}}.
$$
For $1\leq k\leq 20$, define the measure $\mu^k$ on $\bbZ_5$ to be 
\[\mu^k_1=\frac{60+4k}{300}\  \mbox{ and } \ \mu^k_i=\frac{60-k}{300}\ \mbox{ for } 2\leq i\leq 5.\]
We set $N=3000$, and consider the isometry $V_k$ associated with $\SBM(A, \mu^k,N)$ given in \cref{def:M-A-mu-D} \ref{def:V-isom}.
Inspired by \cref{prop:lin-alg-compression} \ref{part-prop:lin-alg:X-eigen->Y-eigen}, we use $V_k$ to transfer the basis $\{\phi_i\}_{i=1}^5$ of $A$ to an orthonormal set $\{\xi_i^k\}_{i=1}^5\subseteq {\mathbb C}^N$ defined as  $\xi_i^k:=V_k\phi_i$. We call $\{\xi_i^k\}_{i=1}^5$ the \emph{transferred character basis} for the range of the model matrix $W_k$ of $\SBM(A, \mu^k,N)$.
Note that the transferred character basis is \emph{not} an eigenbasis of the range of $W_k$ unless $\mu^k$ is uniform. 
We now present a numerical study to compare the transferred character basis with the SBM-basis $\{ y_i^k\}_{i=1}^n$ for $W_k$ derived in \cref{thm: general-decomp}.
 
This provides us with a class of examples of SBMs in which one block is larger than the rest, and the difference in block sizes becomes more pronounced as $k$ increases.

\begin{figure}[ht]
\subfloat[\centering Plot of the inner products between transferred character basis $\{\xi_i^k\}_{i=1}^5$ and the SBM Fourier basis $\{y^k_i\}_{i=1}^5$ for $k=1,\ldots, 20$]{{\includegraphics[width=7cm]{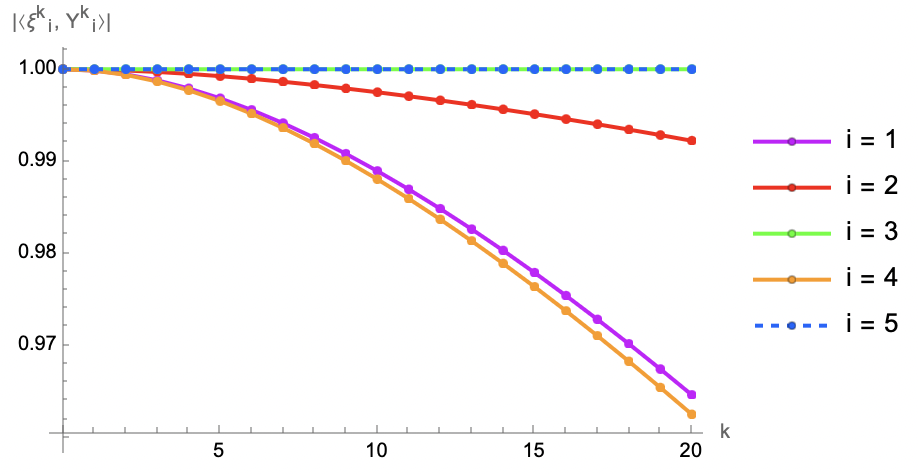} }}
\qquad
\subfloat[\centering Limiting graphon when $k=0$]{{\includegraphics[width=2.5cm]{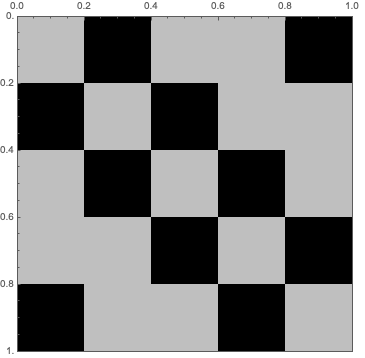} }}%
    \qquad
    \subfloat[\centering Limiting graphon when $k=20$]{{\includegraphics[width=2.5cm]{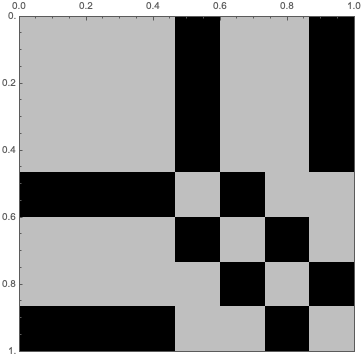} }}

\caption{Comparison of transferred character basis with SBM Fourier basis for non-uniform measures}
    \label{fig:one-large-block-test}
    
\end{figure}

\cref{fig:one-large-block-test} gives the inner product of the corresponding eigenvectors in the two bases. That is, for a given $i$ and $k$ the inner product $\langle \xi_i^k,y_i^k\rangle$ is displayed. We see that the agreement between the transferred character basis and the SBM Fourier basis is quite good, although it does decrease as the block sizes become less uniform. 
 \cref{fig:one-large-block-test} also demonstrates \cref{prop:one-large-block}, since we are in the case of one large block. By this proposition, $\xi_3^k$ and $\xi_5^k$ are eigenvectors of $W_k$, and thus belong to the SBM-driven Fourier basis. 
We see in the figure that indeed  $\xi_3^k$ and  $y_3^k$ are in perfect agreement for  all $k$. The same is true for $\xi_5^k$ and $y_5^k$.

\subsection{Three different block sizes}\label{subsec-3block}
Next, we consider the measure  $\mu^k_1 = \frac{30+2k}{150}$, $\mu^k_2 = \frac{30+k}{150}$ and $\mu^k_3=\mu^k_4=\mu^k_5 = \frac{30-k}{150}$ and take $N=150$. We let $k$ vary from 0 to 20. Again, we compared the transferred character basis, computed as in the previous section, with the SBM  Fourier basis, computed using \cref{thm: general-decomp}. For each $k$, let $\{y^k_i\}_{i=1}^N$ be the eigenvectors of the model matrix $W_k$.
 
As before, we compute the inner product between $\{y^k_i\}_{i=1}^5$ and the corresponding vectors of the transferred character basis $\{\xi^k_i\}_{i=1}^5$.
The results are shown in \cref{fig:decayPlot and modelComparison} (a). At $k=0$, all block sizes are equal so we have perfect agreement. As $k$ increases, the agreement deteriorates more quickly than in the case of just two distinct block sizes.

\begin{figure}[ht]

 \centering
    \subfloat[Plot of the magnitude of the inner product between the transferred character basis $\{\xi^k_{i}\}_{i=1}^5$ and the SBM Fourier basis $\{y^k_{i}\}_{i=1}^5$ for $k=0,\ldots, 20$.]{{\includegraphics[width=7cm]{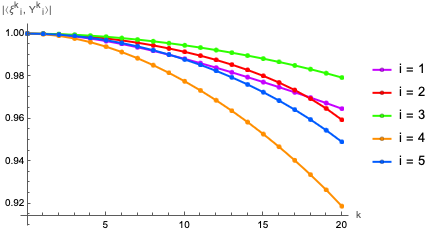} }}%
    \qquad
    \subfloat[ Inner product between the transferred character basis elements $\{\xi^m_i\}_{i=1}^5$ and the corresponding SBM basis $\{y^m_i\}_{i=1}^5$ of model 1 from \cref{sec:VaryingBlockSize} (blue) and model 2 from \cref{subsec-3block} (red).]{{\includegraphics[width=7cm]{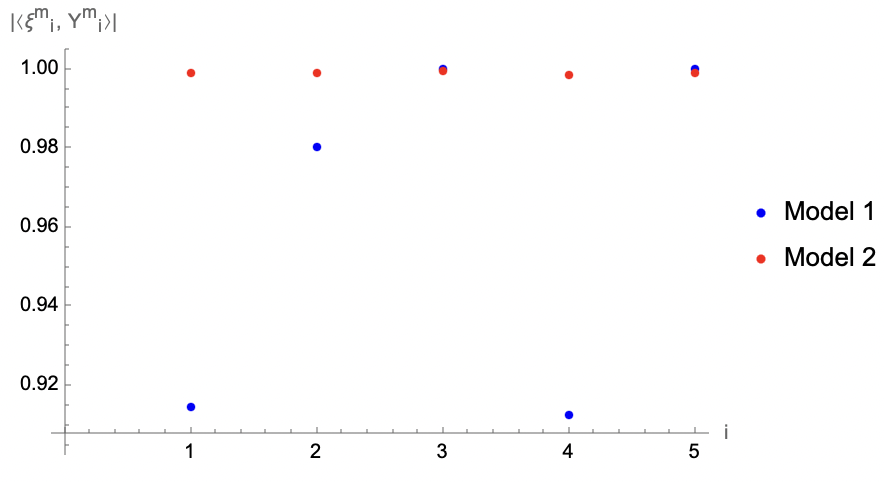} }}%
    \caption{Comparison of graphon with 1 perturbed block versus 2 perturbed blocks}%
    \label{fig:decayPlot and modelComparison}%
    \vspace{-0.2cm}
\end{figure}

In \cref{fig:decayPlot and modelComparison} (b), we compare two different models.  Model 1 is constructed from a measure that gives blocks of size 2000, 250, 250, 250, 250, and Model 2 is constructed from blocks of size 1350, 1344, 1102, 1102, 1102. We observe that  the agreement between the transferred character basis and the actual SBM basis for model 2 is strong. However, for model 1 there is only strong agreement for $\xi_3$ and $\xi_5$. This is despite the relatively large perturbations in the block sizes present in model 1. As we saw earlier, $\xi_3$ and $\xi_5$ are SBM basis elements in  model 1. Model 2 has smaller deviation from the uniform measure, and gives better performance overall.

\section*{Acknowledgements}
MG acknowledges support from NSF Grant DMS-2408008, NSF Grant CCF-2427965 and Simons Travel Support for Mathematicians. 
JJ was funded by an NSERC Discovery Grant. This paper was initiated while two of the authors conducted an RIT
visit at the Banff International Research Center. We are grateful to BIRS for the financial
support and hospitality. 
This material is also based upon work supported by the Swedish Research Council under grant no. 2021-06594 while M.G.~was in residence at Institut Mittag-Leffler in Djursholm, Sweden during the semester ``Operator
algebras and Quantum information theory'' Spring, 2026.
Finally, we sincerely thank the editor and anonymous reviewers for critically reading the manuscript and suggesting numerous improvements.

\newpage
\appendix
\section{Matrix perturbation theory}\label{sec:perturb}

 The relation between the eigenvalues and eigenvectors of a symmetric, real-valued matrix $A$ and a perturbed matrix $\tilde{A}$ has long been the object of study. For an overview, see \cite{StewartSun1990}. In the following, let $A$ and $\tilde{A}$ be symmetric matrices in $\bbR^{n\times n}$, where $H=A-\tilde{A}$ is considered to be the error between the matrix $A$ and its perturbation $\tilde{A}$. Assume that $A$ has eigenvalues $\lambda_1\geq \lambda_2\geq\dots\geq \lambda_n$, and $\tilde{A}$ has eigenvalues $\tilde\lambda_1\geq \tilde\lambda_2\geq\dots\geq \tilde\lambda_n$. Where appropriate, we use `dummy values' $\lambda_0=\infty$ and $\lambda_{n+1}=-\infty$.
 
The difference between eigenspaces of $A$ and $\tilde{A}$ can be bounded via the Davis--Kahan theorem, which is stated in terms of principal angles between subspaces. Let $\ee$ and $\tilde{\ee}$ be $d$-dimensional subspaces of $\bbR^n$, and let $V$ and $\tilde{V}$ be {$n\times d$} matrices whose columns form an orthonormal basis of $\ee$ and $\tilde{\ee}$, respectively. The \emph{principal angles} of $(\ee,\tilde{\ee})$ are the angles $\theta_i=\cos^{-1}(\rho_i)$, $i\in [d]$, where $\rho_1,\rho_2,\dots, \rho_d$ are singular values of $V^\t\tilde{V}$. Let $\Theta={\rm diag}(\theta_1,\theta_2,\dots ,\theta_d)$ be the matrix representing the principal angles, and let $\sin(\Theta)$ be the matrix obtained from $\Theta$ by taking the sine component-wise. The distance between $\ee$ and $\tilde{\ee}$ is given by 
\begin{equation*}\|\sin(\Theta(\ee,\tilde{\ee}))\|_F,
\end{equation*} 
where $\|\cdot \|_F$ denotes the Frobenius norm. 
We now give a bound for this distance. The bound is taken from \cite{YuWangSamWorthDavisKahan2015} and is a variant of the Davis--Kahan theorem. This bound only requires the gap in eigenvalues of one of the matrices, making it suitable for a context where one of the matrices has known properties, and the second matrix is a perturbation of the first.

\begin{theorem}[\cite{YuWangSamWorthDavisKahan2015}, Davis--Kahan]\label{thm:daviskahan}

Let $A$ and $\tilde{A}$ be two symmetric matrices in $\bbR^{n\times n}$, with eigenvalues $\lambda_1\geq \lambda_2\geq \dots \geq \lambda_n$ and $\lambda'_1\geq \lambda_2'\geq \dots \geq \lambda'_n$, respectively. 
Fix $1\leq r\leq s\leq n$ and assume that $\lambda_r=\lambda_{r+1}=\dots =\lambda_s=\lambda$. Assume $\min\{ \lambda_{r-1}-\lambda,\lambda-\lambda_{s+1}\}>0$. Let $d=s-r+1$ be the multiplicity of $\lambda$. Let $\ee$ be the eigenspace of $A$ corresponding to $\lambda$, and let $\tilde{\ee}$ be the space spanned by orthogonal eigenvectors of $\tilde{A}$ corresponding to eigenvalues $\tilde\lambda_r,\tilde\lambda_{r+1},\dots ,\tilde\lambda_{s}$. 
Then
\begin{equation*}
   \|\sin(\Theta(\ee,\tilde{\ee}))\|_F\leq \frac{2\min\{d^{1/2}\|A-\tilde{A}\|_{\opr},  \|A-\tilde{A}\|_F\}}{\min\{ \lambda_{r-1}-\lambda,\lambda-\lambda_{s+1}\}}
\end{equation*}
\end{theorem}

An alternative way of representing the distance between the subspaces, which is relevant in the context of the graph and graphon Fourier transforms, is by quantifying the difference in the projection operators on the two subspaces. Precisely, let $P_{\ee}$ and $P_{\tilde{\ee}}$ be the projection operators onto $\ee$ and $\tilde{\ee}$, respectively. Note that $P_{\ee}$ and $P_{\tilde{\ee}}$ are represented by the matrices $VV^\t$ and $\tilde{V}{\tilde{V}}^\t$, respectively. We can represent the distance between $\ee$ and $\tilde{\ee}$ in terms of the Frobenius distance between these matrices. 

Let $\Theta=\Theta(\ee,\tilde{\ee})$. Then $\cos(\Theta)$ is the diagonal matrix containing the singular values of $V^\t\tilde{V}$. Using the singular value decomposition, we have that there exist orthogonal matrices $U_1,U_2$ so that $V^\t\tilde{V}=U_1\cos(\Theta)U_2^\t$. Now let $W\in \bbR^{n\times n-d}$ be the matrix whose columns form an orthonormal basis of the space orthogonal to $\tilde{\ee}$. Thus 
$WW^\t=I-\tilde{V}\tilde{V}^\t$ is the projection onto $\tilde{\ee}^\perp$. We then have that
$$
V^\t WW^\t V+V^\t\tilde{V}\tilde{V}^\t V=V^\t V=I_d,
$$
since $V$ is orthogonal. Now $V^\t\tilde{V}(\tilde{V}^\t V)=U_1\cos(\Theta)^2U_1^\t$ and thus
$$
(V^\t W)(V^\t W)^\t =I_d-U_1\cos(\Theta)^2U_1^\t =U_1\sin^2(\Theta)U_1^\t.
$$
So we get
\begin{equation*}
    \|\sin(\Theta(\ee,\tilde{\ee}))\|^2_F=\| V^TW\|_F^2=\| P_{\ee}(I-P_{\tilde{\ee}})\|^2_F .
\end{equation*}
Similarly, we have that $\|\sin(\Theta(\ee,\tilde{\ee}))\|^2_F=\| (I-P_{\ee})P_{\tilde{\ee}}\|^2_F$. Putting these two expressions together with the Davis--Kahan theorem we obtain the following corollary.
\begin{corollary}\label{DavisKahanProjections}
 \begin{equation}
    \| P_\ee-P_{\tilde{\ee}}\|_F\leq \frac{2^{3/2}\min\{d^{1/2}\|A-\tilde{A}\|_{\opr},  \|A-\tilde{A}\|_F\}}{\min\{ \lambda_{r-1}-\lambda,\lambda-\lambda_{s+1}\}}.
\end{equation}   
\end{corollary}

\bibliographystyle{amsplain}
\def\cprime{$'$} \def\cprime{$'$} \def\cprime{$'$}
\providecommand{\bysame}{\leavevmode\hbox to3em{\hrulefill}\thinspace}
\providecommand{\MR}{\relax\ifhmode\unskip\space\fi MR }
\providecommand{\MRhref}[2]{%
  \href{http://www.ams.org/mathscinet-getitem?mr=#1}{#2}
}
\providecommand{\href}[2]{#2}

\end{document}